\documentclass[default]{interact}
\usepackage{graphicx}
\usepackage{natbib}
\usepackage{amsmath}
\usepackage{amssymb}
\usepackage{enumitem}
\usepackage{url} 
\usepackage{booktabs,bookmark}
\usepackage{caption}
\usepackage{verbatim}
\usepackage{multirow}
\usepackage{comment}
\usepackage[utf8]{inputenc}
\usepackage{bm}
\usepackage{mathtools}
\usepackage{xspace}
\usepackage{setspace}
\usepackage[caption=false]{subfig}
\usepackage{amsmath}
\usepackage{amssymb}
\usepackage{amsfonts}
\usepackage{multirow}
\usepackage{amsthm}

\theoremstyle{plain}
\newtheorem{theorem}{Theorem}[section]

\newtheorem{assumption}[theorem]{Assumption}
\newtheorem{corollary}[theorem]{Corollary}

\theoremstyle{definition}

\theoremstyle{remark}

\begin{document}

\articletype{RESEARCH ARTICLE}

\title{A new Kernel Regression approach for Robustified  $L_2$ Boosting}

\author{\name{Suneel Babu Chatla\textsuperscript{a} \thanks{CONTACT Suneel Babu Chatla. Email: sbchatla@utep.edu}}
\affil{ Department of Mathematical Sciences, University of Texas at El Paso, El Paso, Texas,  USA }
}
%



\begin{abstract}
We investigate $L_2$ boosting in the context of kernel regression. Kernel smoothers, in general, lack appealing traits like symmetry and positive definiteness, which are critical not only for understanding theoretical aspects but also for achieving good practical performance. We consider a projection-based smoother \citep{hua08analysis} that is symmetric, positive definite, and shrinking. Theoretical results based on the orthonormal decomposition of the smoother reveal additional insights into the boosting algorithm. In our asymptotic framework, we may replace the full-rank smoother with a low-rank approximation.  We demonstrate that the smoother's low-rank ($d_n$) is bounded above by $O(h^{-1})$, where $h$ is the bandwidth. Our numerical findings show that, in terms of prediction accuracy,  low-rank smoothers may outperform full-rank smoothers. Furthermore, we show that the boosting estimator with low-rank smoother achieves the optimal convergence rate. Finally,  to improve the performance of the boosting algorithm in the presence of outliers, we propose a novel robustified boosting algorithm which can be used with any smoother discussed in the study. We investigate the numerical performance of the proposed approaches using simulations and a real application. 
\end{abstract}



\maketitle

\section{Introduction}\label{sec:intro}
Boosting, also known as Gradient boosting, is a prominent machine learning method. It attempts to produce more accurate predictions by integrating the predictions of several ``weak" models which are referred to as  \emph{weak learners} \citep{schapire1990strength, freund1995boosting,freund1996experiments}. Boosting begins with a reasonable estimator, the learner, and  improves iteratively depending on the performance on  training data \citep{buhlmann2003boosting}. Given its empirical success, several attempts have been made by  both the statistics and machine learning communities  to demystify the boosting method's better performance and resistance to overfitting \citep{bartlett1998boosting,breiman1998arcing,breiman1999prediction,schapire1999improved,friedman2000additive,friedman2001greedy,buhlmann2003boosting,park2009}. 

Theoretical underpinnings and practical implementations of boosting are predicated on the assumption that it may be thought of as a functional gradient descent algorithm \citep{breiman1999prediction}. For example, under appropriate risk functions, AdaBoost \citep{freund1996experiments} and LogitBoost \citep{friedman2000additive} algorithms may be seen as optimization problems. \citet{friedman2001greedy} proposed least-squares boosting ($L_2$ boosting), a computationally simple variation of boosting,  and explored some robust algorithms that used regression trees as weak learners. Recent literature on boosting focuses on proposing robust algorithms for classification, regression, and nonparametric regression. For more details, we refer to \citet{ju2021robust}, \citet{li2018boosting}, \citet{miao2015rboost} and the references therein. However,  $L_2$ boosting has received little attention in the literature of kernel smoothing.

The purpose of this research is to give additional insights about $L_2$ boosting utilizing the kernel smoothing framework. We primarily examine the theoretical findings for boosting estimates using a low-rank smoother, which allows the approach to be scaled to huge datasets. In addition, we present a robustified boosting approach to mitigate the influence of outliers in the estimation. These findings are novel in the literature, particularly in the context of kernel smoothing framework. 

We investigate a univariate nonparametric regression model in this study because of it's simplicity in theoretical arguments.  Assume $(X_i,Y_i)$, $i=1,\ldots,n$, are $n$ independent copies of a random pair $(X,Y)$. We consider the model 
\begin{align}
    Y_i &= m(X_i) + \epsilon_i, \qquad i=1,\ldots,n, \label{eqn:model}
\end{align}
 where $m(\cdot)=E(Y \vert X=\cdot)$ is the regression function and $\epsilon_i$'s are random variables with $E(\epsilon_i)=0$ and $\text{var}(\epsilon_i)=\sigma^2$. \citet{buhlmann2003boosting} investigated the properties of $L_2$ boosting and presented expressions for average squared bias and average variance of a boosting estimate for the model (\ref{eqn:model}). These expressions involve eigenvalues and eigenvectors of the corresponding smoother. When these eigenvalues are between 0 and 1, the squared bias decays exponentially quickly and the variance increases exponentially small as the number of boosting iterations increases \citep{buhlmann2003boosting}. The number of boosting iterations is treated as a tuning or regularization parameter in this exponential bias-variance trade-off in the literature. In addition, \citet{buhlmann2003boosting} shows that $L_2$ boosting with smoothing splines achieves optimal rate  $n^{-2\pi/(2\pi+1)}$  if the iteration number $b$ is of order $O(n^{2r/(2\pi+1)})$ as the sample size $n$ goes to infinity, where $r$ is the order of the smoothing spline, and $\pi (> r)$ is the smoothness index of the regression function. 

As indicated in \citet{di2008boosting}, one exciting aspect of boosting  is that it may be utilized as a bias reduction technique, particularly in the kernel smoothing framework. This may be traced back to \citet{tukey1977exploratory}, which refers to one-step boosting as ``twicing''. In case of fixed equispaced $n$ design points, $i/n$, $i=1,\ldots,n$, twicing a kernel smoother is asymptotically equivalent to applying a higher-order kernel \citep{stuetzle1979some}.  \citet{di2008boosting} proposed Nadaraya-Watson $L_2$ boost algorithm and demonstrated its empirical performance. Because their smoother is not symmetric,  it does not provide positive characteristic roots for several popular kernels such as Epanechnikov, Biweight, and Triweight. If the smoother's eigenvalues are outside of $(0, 1]$, boosting will not operate effectively \citep{buhlmann2003boosting}. As a result, their method works effectively only for kernels with strictly positive eigenvalues, such as Gaussian and Triangular. \citet{di2008boosting} also demonstrates that their estimator achieves bias reduction after the first boosting iteration while keeping the variance order  asymptotically the same. \citet{park2009} also investigates Nadaraya-Watson $L_2$ boosting and shows that if the iteration number $b$ is big enough and the bandwidth is appropriately set, $h=O(n^{-1/(2\pi+1)})$, the boosting estimate achieves the optimal rate of convergence.

Low-rank matrix approximation is widely studied in the literature, especially in machine learning and statistics.  It is a popular technique in massive data analysis. For details, we refer to \citet{kishore2017literature} and the references therein. Let $A$ be $n \times n$ real, symmetric matrix, then we write its rank $d_n (<n)$ (low-rank) approximation as 
\begin{align*}
  A_{n \times n} \approx B_{n \times d_n} B_{d_n \times n}^T,
\end{align*}
where $B$ is a rank $d_n$ matrix. This approximation is very economical for storage as it requires only $nd_n$ elements to be stored instead of the original $n^2$ elements. Moreover, there exist efficient probabilistic algorithms for constructing the above decomposition even when $n$ is very large. We refer to \citet{halko2011finding} for an in-depth discussion on this topic. Besides the computational simplicity, the low-rank approximation can also be used to remove noise in the data. In applications, the rank to be removed often corresponds to the noise level where the signal-to-noise ratio is low \citep{chu2003structured}. Boosting is known to produce superior performance, at least empirically, if the base learner is weak\citep{buhlmann2003boosting}. The low-rank smoother serves as a weaker learner than the full rank smoother and hence it may outperform the full rank smoother as we observed in our numerical results in Section \ref{sec:sim}. To the best of our knowledge, our study is the first to investigate the role of low-rank smoothers in boosting.  

In this work, we choose a projection-based smoother \citep{hua08analysis} that is symmetric, positive definite, shrinking and has eigenvalues in $(0, 1]$ for most of the standard kernel functions including Epanechnikov and Biweight. With this smoother, we obtain the asymptotic results for the boosting estimate. Our findings, which are comparable to \citet{buhlmann2003boosting},  are based on the orthonormal decomposition of the smoother\citep{hua08analysis} and offers new insights into the boosting. Because of our asymptotic framework, we can replace the high-rank smoother with a low-rank approximation and execute the boosting approach on a large sample size without having to store all of the smoother's elements. In Theorem \ref{thm:eigen-bias}, we show that the smoother's low-rank ($d_n$) is bounded above by $O(h^{-1})$ where $h$ is the bandwidth. We also give an expression for the approximation error and demonstrate that it converges to zero when the low-rank $d_n=O(h^{-1})$. According to our numerical study, we demonstrate that low-rank smoothers may outperform full-rank smoother in terms of prediction accuracy while lowering the computational effort. Furthermore, in Theorem \ref{thm:min-max}, under some regularity conditions,  we show that boosting estimate of a low-rank smoother achieves the optimal convergence rate, $n^{-2\pi/(2\pi+1)}$, for an appropriately chosen bandwidth, $h=O(n^{-1/(2\pi+1)})$, given the smoothness index $\pi$.

The $L_2$ boosting algorithm, which uses a $L_2$ loss function, is sensitive to outliers in data. Robustification of the boosting method has not received much attention in the literature.  \citet{friedman2001greedy} presented a few robust boosting algorithms by employing regression trees as base learners. Similarly, \citet{lutz2008robustified} presented five robustification algorithms  for linear regression using $L_2$ boosting. We present a novel robust boosting algorithm to estimate the model (\ref{eqn:model}). The proposed method uses a pseudo-outcome approach \citep{cox1983asymptotics,oh2007role, oh2008recipe}, which converts the original problem of robust loss function optimization to the problem of least-squares loss function optimization. This approach is employed at each iteration of the boosting algorithm. We further demonstrate in Theorem \ref{thm:rob-eq} that the estimate based on the pseudo-outcome is asymptotically equivalent to the estimator obtained directly by optimizing the robust loss function. This robustified boosting algorithm is very general and can be used with any smoother. In our numerical study, in addition to the projection-based smoother, we employed this algorithm to both Nadaraya-Watson smoother and spline smoother which is also new.

The paper is organized as follows. In Section \ref{sec:background}, we provide a brief discription of the methods used in the study. Specifically, Section \ref{sec:back} provides a brief introduction to the projection-based smoother matrix \citep{hua08analysis} and Section \ref{sec:l2boost} outlines the algorithm for $L_2$ boosting. In Section \ref{sec:theory} we discuss the theoretical results related to boosting with kernel regression and in section \ref{sec:low-rank} we discuss the asymptotic properties of the boosting estimate of a low-rank smoother. In Section \ref{sec:rob-l2} we outline the robustified boosting algorithm. Section \ref{sec:sim} discusses the simulation results and Section \ref{sec:real-app} illustrates the usefulness of the proposed methods using the data from a real application. We summarize our findings in Section \ref{sec:sc}.

\section{Background} \label{sec:background}
\subsection{ Smoother Matrix } \label{sec:back}
 We consider the local linear modeling approach \citep{fan2018local} which estimates the regression function $m(x)$ in model (\ref{eqn:model}) using a first-order Taylor expansion $m(x)+ m^{(1)}(x)(X-x)$ for $X$ in a neighborhood of $x$. Let 
\begin{align*}
   \bm{X}_x^T &= \begin{bmatrix}
       1 & 1 & \cdots & 1 \\  X_{1}-x & X_2-x & \cdots & X_n-x
   \end{bmatrix}   
\end{align*}
be a design matrix and  $\bm{W}_{x}=\text{diag}\{K_{h}(X_{1}-x),\ldots,K_{h}(X_{n}-x)\}$ be a weight matrix with $K_{h}(\cdot)=K(\cdot/h)/h$ where $K(\cdot)$ is a positive and symmetric probability density function defined on a compact support, say $[-1,1]$, and  $h$ is a bandwidth. Denote the response vector $\bm{y}=(Y_1,\ldots,Y_n)^T$ and coefficient vector $\bm{\beta}=(\beta_0, \beta_1)^T$. Then 
\begin{align}
(\widehat{m}(x), \widehat{m}^{(1)}(x)) &= \underset{\beta_0, \beta_1}{\text{min}} \frac{1}{n} \sum_{i=1}^n \left(Y_i- \beta_0 - \beta_1 (X_{i}-x) \right)^2 K_{h}(X_{i}-x) \nonumber \\  
&=  \underset{\beta_0, \beta_1}{\text{min}} \frac{1}{n} \left(\bm{y}-\bm{X}_{x}\bm{\beta}\right)^T \bm{W}_{x}\left(\bm{y}-\bm{X}_{x}\bm{\beta}\right).  \label{eqn:back-uob}
\end{align}
Suppose the random variable $X$ is compactly supported, say $[0, 1]$. Let $K_h(u,v)$ be the boundary corrected kernel defined in \citet{mammen1999existence} as  
\begin{align*}
  K_h(u,v) &= \frac{K_h(u-v)}{\int K_h(w-v)dw} I(u,v \in [0, 1]),
\end{align*}
where $I(\cdot)$ is an indicator function.
The smoother matrix $\bm{H}_{1}^*$ (for local linear) in \citet{hua08analysis} is based on integrating local least-squares errors (\ref{eqn:back-uob})
\begin{align}
\frac{1}{n} & \int \sum_{i=1}^n \left(Y_i-\widehat{m}(x)- \widehat{m}^{(1)}(x)(X_{i}-x)\right)^2 \nonumber \\ & \qquad \qquad \times K_{h}(x, X_{i}) dx
=   \frac{1}{n} \bm{y}^T\left(\bm{I}-\bm{H}_{1}^*\right)\bm{y},
\label{eqn:back-iuobj}
\end{align}
where $\bm{I}$ is an $n$- dimensional identity matrix and 
\begin{align*}
  \bm{H}_1^* &= \int \bm{W}_x \bm{X}_x (\bm{X}_{x}^{T}\bm{W}_{x}\bm{X}_{x})^{-1} \bm{X}_x^T \bm{W}_x dx  
\end{align*}
 is a local linear smoother with its $(i,j)$th element, $w_j(X_i)$, is defined as
\begin{align}
  w_j(X_i) &=\int \big[K_{h}(x, X_{i}) \quad (X_i-x) K_{h}(x, X_{i})\big](\bm{X}_{x}^{T}\bm{W}_{x}\bm{X}_{x})^{-1} \nonumber\\ & \qquad \qquad \qquad \times \big[K_{h}(x, X_{j}) \quad (X_j-x) K_{h}(x, X_{j})\big]^T dx.
\label{eqn:back-hdef1}
\end{align}
It is worth mentioning that $\sum_{j=1}^nw_j(X_i)=1$, $\sum_{i=1}^nw_j(X_i)=1$, and $\sum_{j=1}^n w_j^2(X_i)=O_p(n^{-1}h^{-1})$ for $i=1,\ldots,n$. 
The $i$th fitted value can be written as 
\begin{align}
  \label{eqn:back-fit} 
\widehat{m}_{LL}^*(X_{i})= \bm{e}_i^T \bm{H}_1^* \bm{y} = \sum_{j=1}^n w_j(X_i) Y_j &=\int \left(\widehat{m}(x)+\widehat{m}^{(1)}(x) (X_{i}-x) \right) K_{h}(x,X_{i})dx,
\end{align}
where $\bm{e}_i$ is the unit vector of length $n$ with 1 at the $i$th position. In other words,
\begin{align}
    \widehat{\bm{m}}_{LL}^* := (\widehat{m}_{LL}^*(X_{1}),  \ldots, \widehat{m}_{LL}^*(X_{n}))^T = \bm{H}_1^* \bm{y}. \label{eqn:back-hfit}
\end{align}
The estimator $\widehat{m}_{LL}^*(X_{i})$  in  (\ref{eqn:back-fit}) achieves bias reduction at interior points $[2h,1-2h]$ \citep{he2009double}. For this reason, the bias of  $\widehat{m}_{LL}^*(X_{i})$ is of order $h^{4}$ instead of $h^2$ which is the standard order of bias for local linear estimator $\widehat{m}(x)$. \citet{hua08analysis} and \citet{hua14local} show that $\bm{H}_{1}^*$ is symmetric,  positive definite, and shrinking. 

Similar arguments can be used to define the smoother $\bm{H}_{0}^*$ (local constant) and the $i$th fitted value $\widehat{m}_{LC}(X_i)$ for local constant modeling. The estimator $\widehat{m}_{LC}^*(X_i)$ takes the following form
\begin{align}
  \label{eqn:back-fit-lc} 
\widehat{m}_{LC}^*(X_{i})&= \bm{e}_i^T \bm{H}_0^* \bm{y} =\int \widehat{m}_{NW}(x)  K_{h}(x,X_{i})dx,
\end{align}
where $\widehat{m}_{NW}(x)$ is the Nadaraya-Watson estimator.  The estimator $\widehat{m}_{LC}^*(X_{i})$ has the bias of order $h^2$ at interior points.


\subsection{ Boosting with Kernel Regression}\label{sec:l2boost}

The $L_2$ boosting algorithm may be considered as an application of the functional gradient descent technique, as shown in \citet{buhlmann2003boosting} and \citet{park2009}.  We present a pseudo algorithm for estimating $m(\cdot)$ in (\ref{eqn:model}). Hereafter, notation $\bm{H}^*$ is used to denote both $\bm{H}_{1}^*$ and $\bm{H}_{0}^*$. Denote $\widehat{\bm{m}}_a^*=(\widehat{m}_a^*(X_1),\ldots, \widehat{m}_a^*(X_n))^T$ for a given subscript $a$.

\hspace*{2em}

\textbf{(I). $L_2$ Boosting Algorithm:}

\begin{itemize}
    \item[] Step 1 (Initialization): Given the data $\{X_i, Y_i\}$, $i=1,\ldots,n$, fit an initial estimate $\widehat{\bm{m}}^*_0:= \bm{H}^*\bm{y}$ based on (\ref{eqn:back-hfit}) and (\ref{eqn:back-fit-lc}).
    \item[] Step 2 (Iteration): Repeat for $b=1,\ldots,B$.
    \begin{enumerate}
        \item Compute the residual vector $\bm{\delta} :=\bm{y}-\widehat{\bm{m}}^*_{b-1}$.
        \item Fit a nonparametric regression model to residual vector $\bm{\delta}$ to obtain $\widehat{\bm{m}}^*_\delta:= \bm{H}^*\bm{\delta}$.
        \item Update $\widehat{\bm{m}}^*_b = \widehat{\bm{m}}^*_{b-1} + \widehat{\bm{m}}^*_\delta$.
    \end{enumerate} 
\end{itemize}
\hspace{2em}

\citet{di2008boosting} and  \citet{park2009} investigated the properties of $L_2$ boosting for Nadaraya - Watson smoother, $\bm{S}_{NW}= ((K_h(X_{i}-X_j)/\sum_{j=1}^n K_h(X_i-X_j) ))_{1 \le i,j \le n}$.  As shown in \citet{buhlmann2003boosting}, one of the requirements for bias reduction with boosting is to have eigenvalues are in $(0,1]$. The smoother $\bm{S}_{NW}$, on the other hand, is not symmetric and may not contain all of the eigenvalues in $(0,1]$. The requirements on kernel functions for which $\bm{S}_{NW}$ yields eigenvalues in $(0,1]$ are provided by \citet{cornillon2014recursive}. They prove that the spectrum of $\bm{S}_{NW}$ ranges between zero and one if the inverse Fourier-Stieltjes transform of a kernel is a real positive finite measure. While this is true for positive definite kernels like Gaussian and Triangular, it is not true for non-positive definite kernels like Epanechnikov and Uniform. Given that kernels are chosen considering their computational efficiency, this finding is unexpected. Smoothers using Epanechnikov kernels converged quicker than those with Gaussian kernels in our numerical analysis in Section \ref{sec:sim}. 

The eigenvalues for both $\bm{S}_{NW}$ and $\bm{H}^*$ smoothers are shown in Figure \ref{fig:eig-h-Epan-Gaus} across  different bandwidth values in $(0,1]$. The eigenvalues of a cubic spline smoother for different values of smoothing parameters are also presented for comparison. Epanechnikov and Gaussian kernels are used in the first three graphs in the top and bottom rows, respectively. While $\bm{S}_{NW}$ produces some negative eigenvalues for the Epanechnikov kernel, it produces only positive eigenvalues for the Gaussian kernel. Figure \ref{fig:mse-h-Epan-Gaus} depicts the mean squared error ($MSE$)
\begin{align}
  MSE &= n^{-1} \sum_{i=1}^n \{Y_i-\widehat{m}_b^*(X_i)\}^2,
\label{eqn:mse}
\end{align}
values of boosting for a different number of boosting iterations, $b=0,1\ldots,1000$, for sample size $n=500$. For different $h$ and $\lambda$ (spline smoothing) values $\{0.2,0.4,0.6,0.8,1\}$, the $MSE$ values are averaged over 50 samples. Because of the negative eigenvalues, it appears that boosting does not achieve bias reduction with smoother $\bm{S}_{NW}$ using Epanechnikov kernel. Figure \ref{fig:mse-h-Epan-Gaus} shows this with increasing $MSE$ values in the top left plot.
%
 \begin{figure*}
    \centering
     \begin{minipage}{1\textwidth}
      \centering
      \includegraphics[width=0.9\textwidth,height=0.175\textheight]{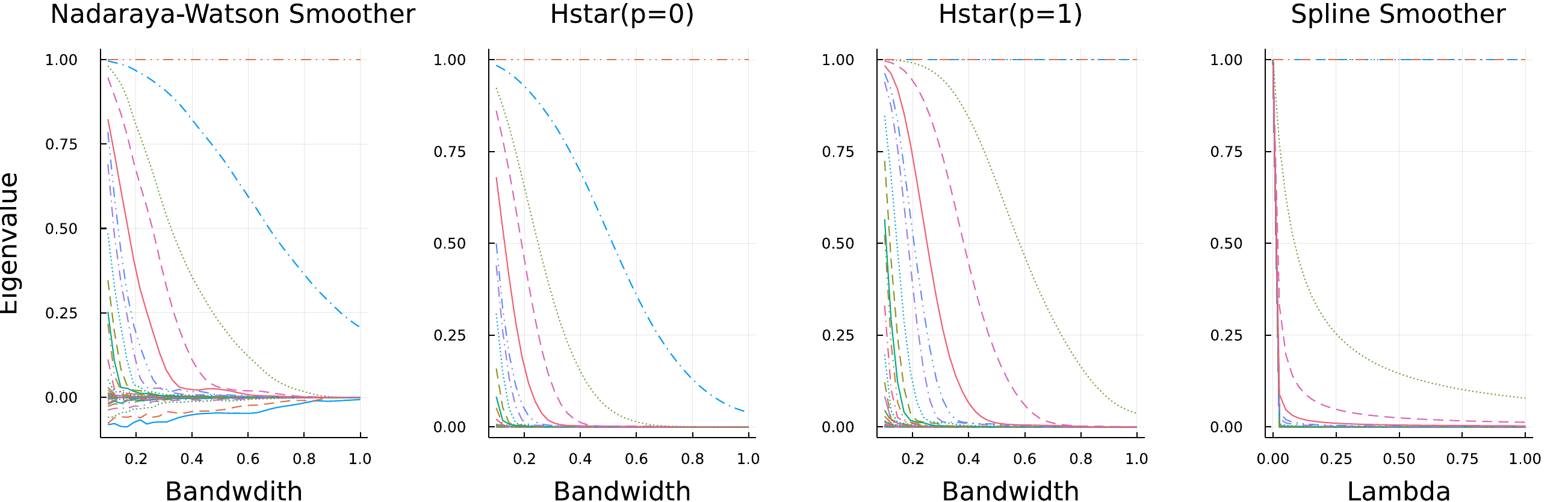}
      \end{minipage}
      \hfill
      \begin{minipage}{1\textwidth}
        \centering
        \includegraphics[width=0.9\textwidth,height=0.175\textheight]{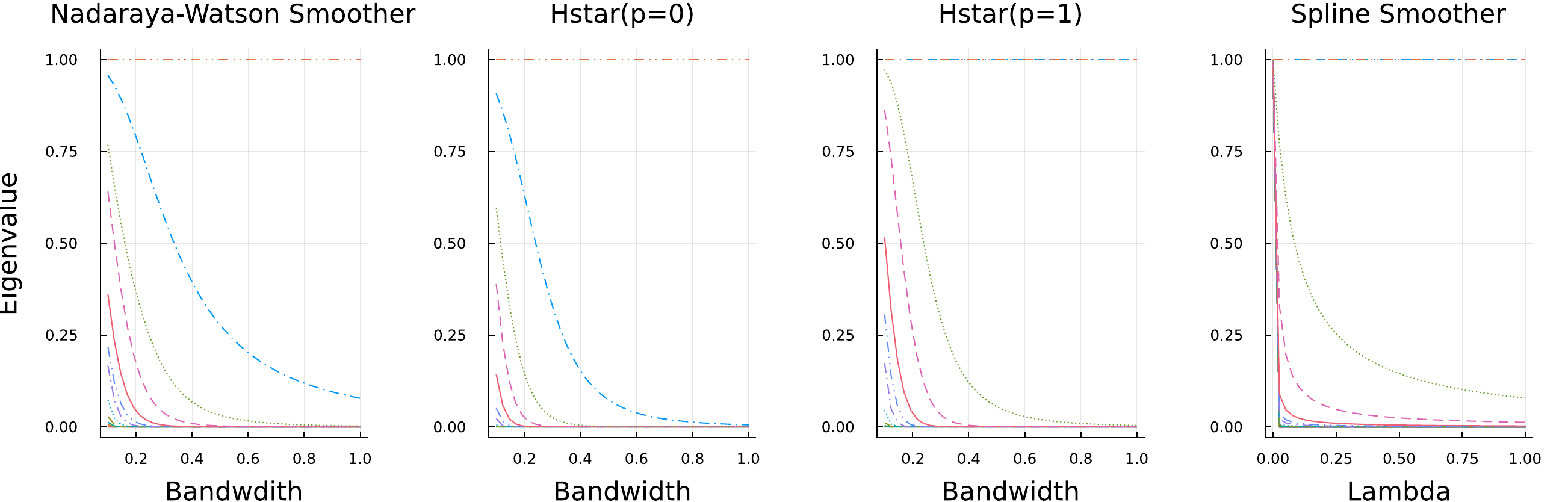}    
      \end{minipage}
    \caption{ Eigenvalues for (left) Nadaraya-Watson and (middle) $\bm{H}^*$ smoothers, and (right) cubic spline smoother. Epanechnikov kernel: top row,  Gaussian kernel: bottom row.  Simulation design is based on the Example 1 in Section \ref{subsec:numeric} with $n=20$}.
    \label{fig:eig-h-Epan-Gaus}
  \end{figure*}

According to \citet{buhlmann2003boosting}, the $L_2$ boost estimate at iteration $b$ of the Algorithm (I) is defined as 
\begin{align}
  \widehat{\bm{m}}_b^* &= \sum_{j=0}^b \bm{H}^*(\bm{I}-\bm{H}^*)^j\bm{y}= \mathcal{S}_b \bm{y}, \label{eqn:mbstar}
\end{align}
where $\mathcal{S}_b = \bm{I}-(\bm{I}-\bm{H}^*)^{b+1}$. The boosting update $\widehat{\bm{m}}^*_b:= \widehat{\bm{m}}^*_{b-1} + \bm{H}^*(\bm{y}-\widehat{\bm{m}}^*_{b-1})$ is expressed as
\begin{align}
  \mathcal{S}_b \bm{y} =  \mathcal{S}_{b-1} \bm{y} + \bm{H}^* (\bm{I}-\mathcal{S}_{b-1}) \bm{y}.
  \label{eqn:bupdate}  
\end{align}
 \begin{figure*}[h]
  \centering
 \begin{minipage}{1\textwidth}
   \centering
   \includegraphics[width=0.9\textwidth,height=0.175\textheight]{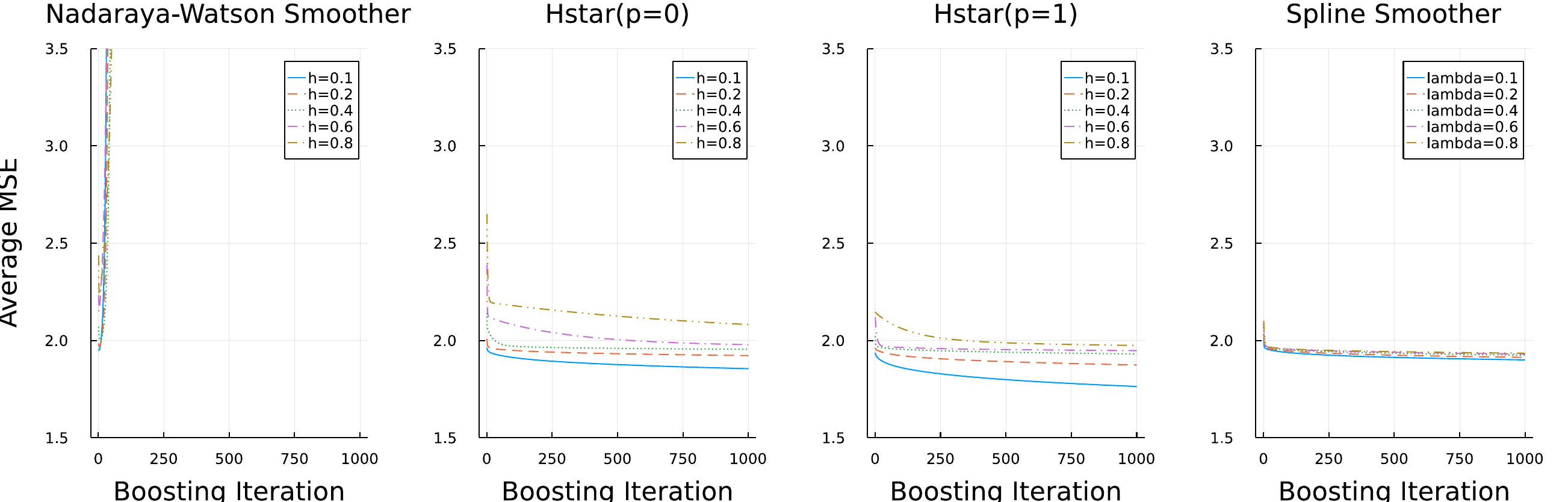}
  \end{minipage}
  \hfill
  \begin{minipage}{1\textwidth}
  \centering
  \includegraphics[width=0.9\textwidth,height=0.175\textheight]{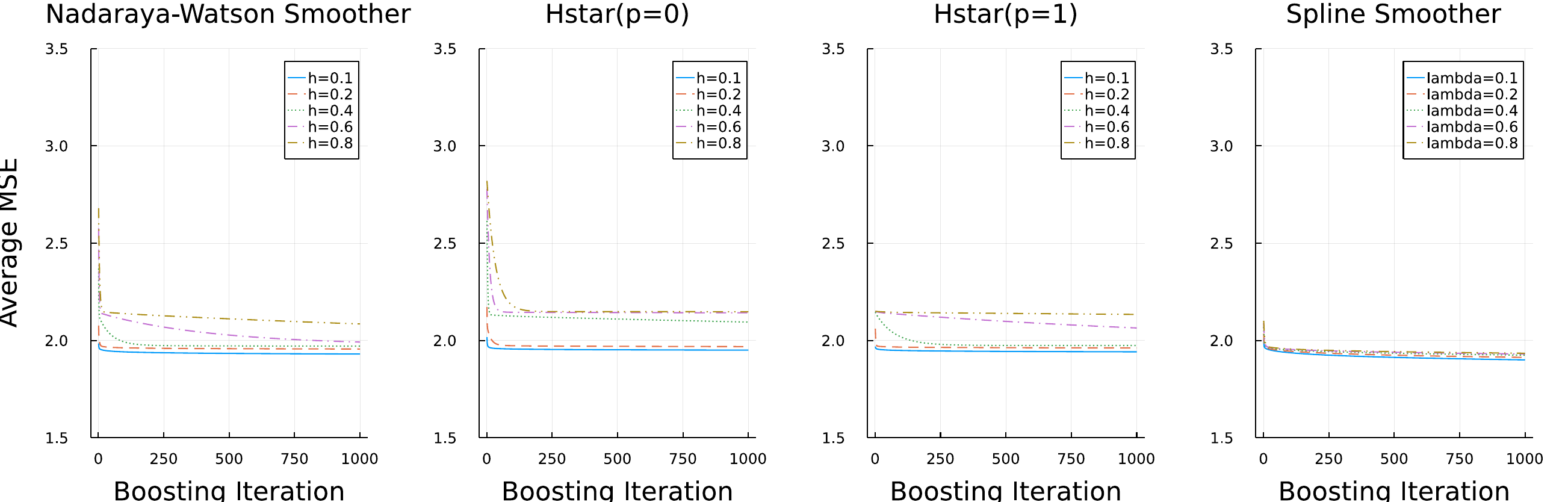}
  \end{minipage}
  \caption{Average of mean squared error (MSE) (\ref{eqn:mse}) values across 50 samples of size $n=500$ for (left) Nadaraya-Watson and (middle) $\bm{H}^*$ smoothers, and (right) cubic spline smoothers. Epanechnikov kernel (top) and  Gaussian kernel ((bottom). Simulation design is based on Model (M1) in Section \ref{subsec:numeric} with $n=500$. }
  \label{fig:mse-h-Epan-Gaus}
\end{figure*}

\section{Contributions to Boosting} \label{sec:theory}

For theoretical examination of the boosting estimator, $\widehat{\bm{m}}_b^*$, $b=0,1,\ldots$, we give the following assumptions.  In the following, the order of the local polynomial $p=0$ represents the results for $\bm{H}_0^*$ smoother and $p=1$ for $\bm{H}_1^*$ smoother.

\begin{assumption}\label{as:1-x}  The covariate 
	density $f(\cdot)$ of $X$  has at least two continuous derivatives on  a compact support, $\mathcal{I}=[0, 1]$, and $\underset{x \in \mathcal{I}}{\inf} f(x) > 0$. 
\end{assumption}
\begin{assumption} \label{as:2-ker-h}  The kernel $K(\cdot)$ is a bounded
	symmetric density function with bounded support $[-1,1]$ and satisfies Lipschitz
	condition. The bandwidth $h \rightarrow 0$   and $n h / (\ln n)^2 \rightarrow \infty$ as $n \rightarrow \infty$.
\end{assumption}
The assumption \ref{as:2-ker-h} is common in the nonparametric smoothing literature. It is a mild condition on the kernel function, for example, please refer to \citet{hua14local} and  \citet{park2009}. Assumptions \ref{as:1-x} is a mild condition on the design density which is required for the projection-based estimators to achieve bias reduction, please see \citet{hua08analysis} and  \citet{hua14local}.

Let $0 < \lambda_k \le 1$, $k=1,\ldots,n$, be the eigenvalues of the smoother $\bm{H}^*$. For our results, we need the asymptotic orders of the eigenvalues of $\bm{H}^*$ which are presented in  Theorem 4 of \citet{hua14local}.  For convenience, we summarize their results in the following corollary. In the following, $p=0$ give results for the eigenvalues $\{\lambda_k\}$ from $\bm{H}_0^*$ matrix in (\ref{eqn:back-fit-lc}) and $p=1$ give results for the eigenvalues $\{\lambda_k\}$ from $\bm{H}_1^*$ matrix in (\ref{eqn:back-hfit}). 
\begin{corollary}[Theorem 4 of \citet{hua14local}]\label{prop:eigen}
    Suppose the Assumptions \ref{as:1-x} and \ref{as:2-ker-h} hold. Then conditioned on $\bm{x}=(X_1,\ldots,X_n)^T$, we have the following results hold for $1 \ge \lambda_{1} \ge \lambda_{2} \ge \ldots \ge \lambda_n >0$  and for $p=0, 1$. 
    \begin{enumerate}[label=(\alph*)]
        \item   Suppose $k \rightarrow \infty$ and $kh \rightarrow 0$. The eigenvalues $\lambda_k=1-(kh)^{2(p+1)}M_k$, for some non negative constant $M_k$ which takes zero when $\lambda_k=1$ and takes some positive value when $0<\lambda_k<1$, and for $k= 1, 2, \ldots$.
        \item Suppose $kh \rightarrow C$ for some constant $C$ or $kh \rightarrow \infty$ as $k \rightarrow \infty$. Then $(kh)^{2(p+1)}M_k \rightarrow 1$ and eigenvalues $\lambda_k \rightarrow 0$.
        \item  The corresponding eigenvectors for $0<\lambda_k<1$ are asymptotically the trigonometric polynomials $\cos(2k\pi \bm{x})$ and $\sin(2k\pi \bm{x})$.
    \end{enumerate} 
\end{corollary}
 According to Theorem 4 in \citet{hua14local} and Figure \ref{fig:eig-h-Epan-Gaus}, some eigenvalues have value 1 and the number of eigenvalues with value 1 is at least $p+1$, for $p=0,1$.  Parts (a) and (b) of the Corollary \ref{prop:eigen} show that  the eigenvalues $0<\lambda_k<1$ converge to 1 for $k$ values that fulfill $kh \rightarrow 0$,  whereas the eigenvalues $\lambda_k$ converge to zero for $kh \rightarrow C$ for constant $C$ or $kh \rightarrow \infty$.


We now derive the asymptotic properties of the proposed boosting estimate in (\ref{eqn:bupdate}). We obtain the orthonormal decomposition of $\mathcal{S}_b$ in (\ref{eqn:bupdate}) as described in Proposition 2 of \citet{buhlmann2003boosting},
\begin{align*}
    \mathcal{S}_b  = \mathcal{U} \mathcal{D}_b \mathcal{U}^T, \qquad   \mathcal{D}_b=\text{diag}\{1-(1-\lambda_k)^{b+1}\},
\end{align*}
where $\mathcal{U}$ includes the orthonormal eigenvectors of $\bm{H}^*$. Denote the true regression function as $\bm{m}$$=(m(X_1),\ldots,m(X_n))^T$, and define $\bm{\gamma}=(\gamma_1,\ldots,\gamma_n)^T= \mathcal{U}^T\bm{m}$.
According to \citet{buhlmann2003boosting} and Corollary \ref{prop:eigen},  the conditional squared bias and the conditional variance for $L_2$ boost estimate are computed as 
\begin{align}
  \text{bias}^2(b;\widehat{m}_b^*) &= n^{-1}\sum_{k=1}^n \left(E[\widehat{m}_b^*(X_k)| X_1,\ldots,X_n] -m(X_k)  \right)^2 \nonumber \\ &= n^{-1} \sum_{k=1}^n \gamma_k^2 \left\{1-\lambda_k \right\}^{2(b+1)} \nonumber \\
            &= n^{-1} \sum_{k=1}^n \gamma_k^2 \left\{(kh)^{2(p+1)} M_k\right\}^{2(b+1)}, \label{eqn:full-bias}
\end{align}
and
\begin{align}
  \text{var}(b;\widehat{m}_b^*) &= \sigma^2 n^{-1} \sum_{k=1}^n \text{Var}(\widehat{m}_b^*(X_k)| X_1,\ldots,X_n) \nonumber\\ &=\sigma^2 n^{-1} \sum_{k=1}^n \left\{ 1- [1- \lambda_k]^{b+1} \right\}^2 \nonumber\\
            &= \sigma^2 n^{-1}  \sum_{k=1}^n \left\{ 1- [(kh)^{2(p+1)} M_k]^{b+1} \right\}^2. \label{eqn:full-var}
\end{align}
\subsection{Boosting with a low-rank smoother} \label{sec:low-rank}
We now consider a low-rank smoother. Let $\bm{H}^*(d)$ be a rank $d_n (<n)$  smoother of $\bm{H}^*$. Without loss of generality, we assume that the eigenvalues $\lambda_k$'s are sorted in non-increasing order ($1 \ge \lambda_{1} \ge \lambda_{2} \ge \ldots \ge \lambda_{d_n} \ge \cdots \ge \lambda_n >0$). Therefore, 
\begin{align*}
  \mathcal{S}_b(d) &=  \mathcal{U}(d)_{n \times d_n} \mathcal{D}_b(d) \mathcal{U}^T(d)_{d_n \times n}, \qquad   \mathcal{D}_b(d)=\text{diag}\{1-(1-\lambda_k)^{b+1}\}_{k=1}^{d_n},  
\end{align*}
 is the corresponding low-rank smoother for $\mathcal{S}_b$. Suppose $\widehat{m}_b^*(d)$ is the boosting estimate at the $b$th iteration. The following theorem provides an upper bound for the low-rank, $d_n$, and expressions for the approximation error of the smoother, $\mathcal{S}_b(d)$,  and the squared bias and variance of the boosting estimator $\widehat{m}_b^*(d)$. 
\begin{theorem} \label{thm:eigen-bias}
    Suppose the Assumptions \ref{as:1-x} and \ref{as:2-ker-h} hold. Then conditioned on $\{X_1,\ldots,X_n\}$, the following results hold for $p=0,1$.
    \begin{enumerate}[label=(\alph*)]
      \item The low-rank $d_n$ of a smoother $\mathcal{S}_b(d)$ is bounded above by $O(h^{-1})$.
      \item The approximation error of the low-rank smoother, $\mathcal{S}_b(d)$, which is defined as 
      \begin{align*}
         \Vert \mathcal{S}_b- \mathcal{S}_b(d) \Vert_F^2 &= \sum_{k=d_n}^n \{1- [(kh)^{2(p+1)}M_k]^{b+1}\}
      \end{align*}
      where $\Vert \cdot \Vert_F$ is the Frobenius norm, goes to zero when $d_n=O(h^{-1})$.
        \item The squared bias and variance for the low-rank $L_2$ boost estimate $\widehat{m}_b^*(d)$ are 
        \begin{align}
            \text{bias}^2\{b;\widehat{m}_b^*(d)\} &=  n^{-1} \sum_{k=1}^{d_n} \gamma_k^2 \left\{(kh)^{2(p+1)} M_k\right\}^{2(b+1)},
            \label{eqn:low-bias}
        \end{align}
        \begin{align}
            \text{var}\{b;\widehat{m}_b^*(d)\} &=  \sigma^2 n^{-1} \sum_{k=1}^{d_n} \left\{ 1- [(kh)^{2(p+1)} M_k]^{b+1} \right\}^2. \label{eqn:low-var}
        \end{align}
      \end{enumerate}
\end{theorem}
\begin{proof}
  Without loss of generality, we assume that the eigenvalues $\lambda_k$'s are sorted in non-increasing order, $1 \ge \lambda_{1} \ge \lambda_{2} \ge \ldots \ge \lambda_{d_n} \ge \cdots \ge \lambda_n >0$.
  \begin{enumerate}[label=(\alph*)]
    \item From Corollary \ref{prop:eigen}, we know that $\lambda_k \rightarrow 0$ for $kh \rightarrow C$ for some constant $C$ or $kh \rightarrow \infty$. Therefore, for all $k > d_n=O(h^{-1})$, $\lambda_k \rightarrow 0$. Hence, $d_n$ is bounded above by $O(h^{-1})$.
    \item The result follows from the singular value decomposition of the smoother $\mathcal{S}_b$. Again from Corollary \ref{prop:eigen}, for all   $k > d_n=O(h^{-1})$, $ (kh)^{2(p+1)} M_k \rightarrow 1$. Hence the approximation error goes to zero.
    \item The result follows from (\ref{eqn:full-bias}) and (\ref{eqn:full-var}).
  \end{enumerate}
\end{proof}

To the best of our knowledge, the result in Theorem \ref{thm:eigen-bias} is new to the literature since it offers  the bias and variance of a boosting estimate using a low-rank smoother. We have the following remarks.
\begin{itemize}
    \item The bias and variance expressions in (\ref{eqn:low-bias}) and (\ref{eqn:low-var}) differ from the standard bias expressions for the kernel smoothing estimators (as shown below), which involve derivatives of the true function $m(\cdot)$, where as the bias expression in (\ref{eqn:low-bias}) includes the Fourier coefficients $\gamma_k$ of the true function $m(\cdot)$. 
    \item While Theorem \ref{thm:eigen-bias} shares similarities with  Proposition 3 in \citet{buhlmann2003boosting}, we may not provide the exponential bias and variance trade-off result as they did in their Theorems 1 and 2.  \citet{buhlmann2003boosting} uses the boosting iteration number $b$ as a regularization parameter while keeping the penalty parameter fixed. However, similar to \citet{park2009}, we treat the bandwidth $h$ as a regularization parameter and fix the iteration number $b$. As a result, our setting does not support the exponential bias and  variance trade-off which requires that the number of boosting iterations, $b$,  goes to infinity.
 \end{itemize}

 For the initial estimator $\widehat{m}_0^*$ in the boosting algorithm (I) for $p=1$, \citet{hua08analysis} provided the following asymptotic bias and variance expressions:
\begin{align}
    \text{bias}(\widehat{m}_0^*(X_i)) &= h^4\left( \frac{\mu_2^2-\mu_4}{4} \right) \bigg\{ m^{(4)}(X_i)  + 2m^{(3)}(X_i) \frac{f^{(1)}(X_i)}{f(X_i)} \nonumber \\
    & \qquad  + m^{(2)}(X_i) \frac{f^{(2)}(X_i)}{f(X_i)}\bigg\} + o_p(h^4), \text{ and }  \label{eqn:h0-bias} \\
    \text{var}(\widehat{m}_0^*(X_i)) &= n^{-1}h^{-1} \kappa_0 \sigma^2/f(X_i)(1+o_p(1)), \nonumber
\end{align}
where $\kappa_0$ involves  convolutions of the kernel function. Now, by proceeding similar to Theorem 2 in \citet{park2009}, we can observe that  $L_2$ boosting improves asymptotic bias by $4$ orders of magnitude at each iteration if the true function $m(\cdot)$ and density $f(\cdot)$ are sufficiently smooth. At the same time variance retains the same asymptotic order. This result is consistent with the findings of \citet{di2008boosting} and \citet{park2009} where bias improves by $2$ orders of magnitude because they only considered $p=0$.

We demonstrate the findings of Theorem \ref{thm:eigen-bias} using simulated data from Section \ref{sec:sim}. Figure \ref{fig:rrproj} shows the average $MSE$ values across 25 training samples and the average of $MSE$ values on a test data for the model (M1) in Section \ref{subsec:numeric} using the low-rank smoother $\bm{H}^*_0(d)$ (local constant). For bandwidths $\{ 0.1, 0.2, 0.4, 0.6, 0.8\}$, the model is fitted using smoother with ranks $ d_n= 2, 5, 10$, and $50$. The average $MSE$ values for $d_n=10$ and $d_n=50$ are similar, showing that the low-rank smoother ($d_n=10$) approximates the full-rank smoother ($d_n=50$). Furthermore, in terms of prediction error on test data, it appears that the low-rank smoother with $d_n=5$ outperforms the full-rank smoother. This is investigated further in Table  \ref{tab:ex1-rr} of Section \ref{subsec:numeric}.
\begin{figure}
  \centering
  \begin{minipage}{0.5\textheight}
  \centering
    \includegraphics[width=1.2\textwidth,height=0.2\textheight]{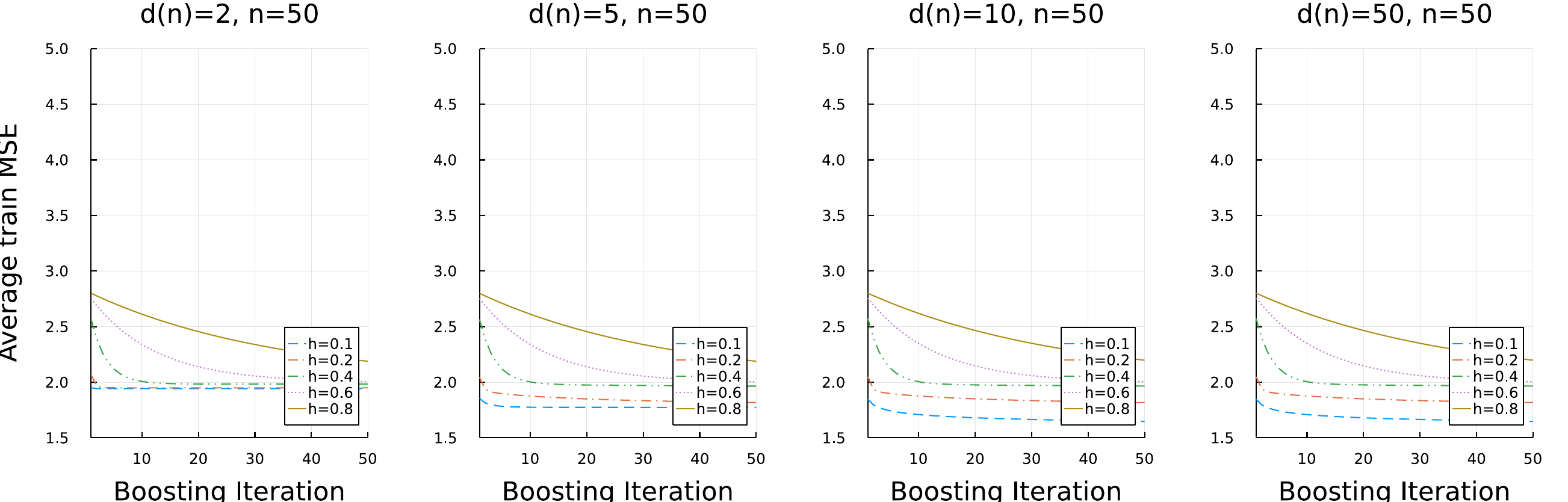}
  \end{minipage}
  \begin{minipage}{0.5\textheight}
    \centering
    \includegraphics[width=1.2\textwidth,height=0.2\textheight]{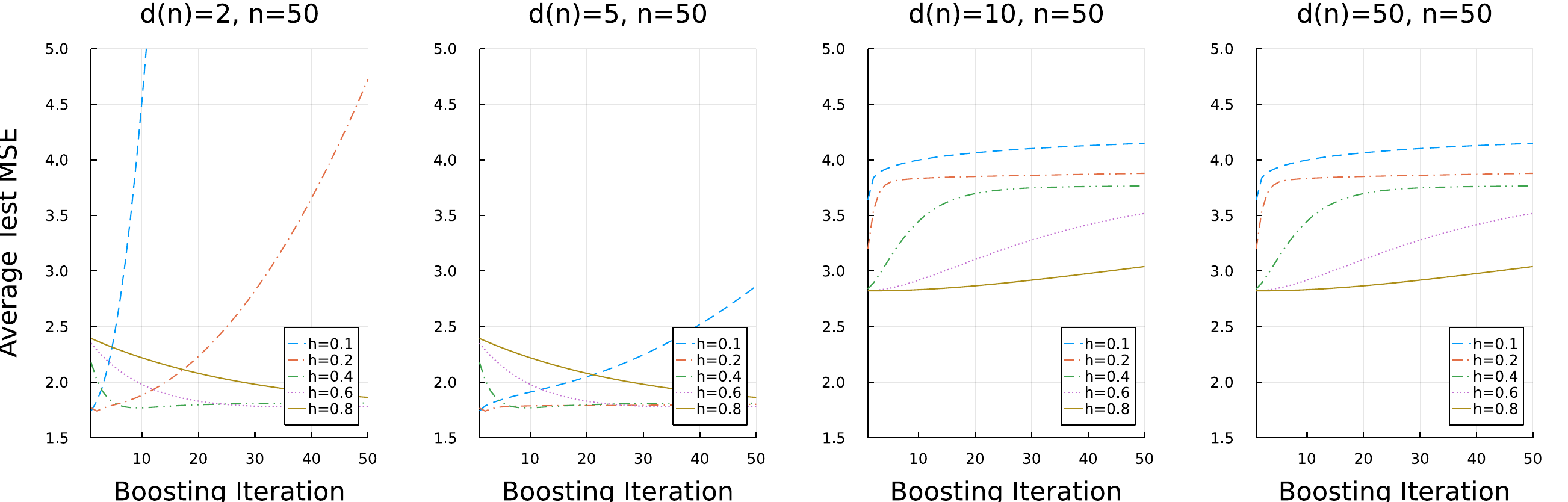}
  \end{minipage}
  \caption{Average of training mean squared error ($MSE$) values (top row) and test $MSE$ values (bottom row) across 25 samples of size $n=50$ for model (M1) in Section \ref{subsec:numeric} using $\bm{H}_0^*$ (local constant) smoother. The first three plots use low-rank $(d_n < n)$ smoother. For the plots in bottom row the same test data is used for all the 25 training samples.  }
  \label{fig:rrproj}
\end{figure}
%

%
We now demonstrate that the $L_2$ boost estimate, $\widehat{m}_b^*(d)$, of a low-rank smoother, $\mathcal{S}_b(d)$, achieves the optimal rate of convergence. First, we define the following Sobolev space of $\pi$th-order smoothness, $\pi \in \mathbb{N}$,
\begin{align}
    \mathcal{F} (\pi) &= \bigg\{ g:  g \text{ is } (\pi-1) \text{ times continuously differentiable and } \nonumber\\ & \qquad \int[g^{(\pi)}(x)]^2 dx < \infty  \bigg\}.
    \label{eq:fclass}
\end{align} 
\begin{theorem}\label{thm:min-max}
    Suppose the Assumption \ref{as:2-ker-h} hold. Assume that the true function $m(\cdot) \in \mathcal{F}(\pi)$ and density $f(\cdot) \in \mathcal{F}(\pi-2)$ for $\pi \ge 4$.  Let $b+1 > \pi/2(p+1)$  for $p=0,1$, and $d_n = O(h^{-1})$ and $d_n \rightarrow \infty$. Then, the $L_2$ boosting  estimate,  $\widehat{m}_b^*(d)$, of a low-rank smoother, $\mathcal{S}_b(d)$, achieves the optimal convergence rate, $n^{-2\pi/(2\pi+1)}$, for the bandwidth $h=O(n^{-1/(2\pi+1)})$.
\end{theorem}
\begin{proof}
  Without loss of generality assume that the eigenvalues are sorted in non-increasing order, $1 \ge \lambda_{1} \ge \lambda_{2} \ge \ldots \ge \lambda_{d_n} \ge \cdots \ge \lambda_n >0$.  By definition, for the true function $m \in \mathcal{F}(\pi)$
  \begin{align*}
      \frac{1}{n} \sum_{k=1}^n \gamma_k^2k^{2\pi} \le M < \infty,
  \end{align*}
  for some constant $M$. First, we bound the bias term in (\ref{eqn:low-bias}). Consider
  \begin{align*}
      \text{bias}^2(b;\widehat{m}_b^*(d)) &= \frac{1}{n} \sum_{k=1}^{d_n} \gamma_k^2 \{ (kh)^{2(p+1)}M_k\}^{2(b+1)}\\
      &=\frac{1}{n} \sum_{k=1}^{d_n} \gamma_k^2 k^{2\pi} \{ (kh)^{2(p+1)}M_k\}^{2(b+1)}k^{-2\pi}\\
      & \le \underset{k=1,\ldots,d_n}{\max} \{ (kh)^{2(p+1)}M_k\}^{2(b+1)}k^{-2\pi} \frac{1}{n}\sum_{k=1}^n \gamma_k^2 k^{2\pi}\\
      &=  \underset{k=1,\ldots,d_n}{\max} g(k) \frac{1}{n}\sum_{k=1}^n \gamma_k^2 k^{2\pi},
  \end{align*}
where $g(k)=\{ (kh)^{2(p+1)}M_k\}^{2(b+1)}k^{-2\pi}$ is an increasing function of $k$ for $b+1 > \pi/2(p+1)$. Since the number of basis functions $d_n$ is of order $O(h^{-1})$, we obtain 
\begin{align*}
  \underset{k=1,\ldots,d_n}{\max} g(k) & \le A h^{2\pi},
\end{align*}
for some constant $A$. Consequently, 
\begin{align*}
  \text{bias}^2(b;\widehat{m}_b^*(d)) &\le A M h^{2\pi},
\end{align*}
which is of order $O(h^{2\pi})$. We now consider the variance term. Since $1-(1-x)^a \le 1-[1-ax]=ax$ for any $x \in [0,1]$ and $a \ge 1$, we have
\begin{align*}
   \text{var}(b; \widehat{m}_b^*(d)) &= \frac{\sigma^2}{n} \sum_{k=1}^{d_n} \left\{ 1- [(kh)^{2(p+1)} M_k]^{b+1} \right\}^2\\
   &\le \frac{\sigma^2 (b+1)^2}{n} \sum_{k=1}^{d_n} (1- (kh)^{2(p+1)} M_k)^2\\
   &\le \frac{\sigma^2 (b+1)^2}{n} \sum_{k=1}^{d_n} (1- (kh)^{2(p+1)} M_k),
\end{align*}
where the last step follows because $\sum_{k=1}^{d_n} a_k^2 \le \sum_{k=1}^{d_n} a_k
$ for $0 < a_k < 1$. Since $d_n$ is $O(h^{-1})$, 
\begin{align*}
  \text{var}(b; \widehat{m}_b^*(d)) &= O(n^{-1}h^{-1}).
\end{align*}
Consequently,
\begin{align*}
 E[\text{MSE}(b; \widehat{m}_b^*(d))| X_1,\ldots, X_n] &\le O(h^{2\pi}) + O(n^{-1}h^{-1}),
\end{align*}
for  $b+1 > \pi/2(p+1)$. Hence the result is proved.
\end{proof}
We have the following remarks:
\begin{enumerate}[label=(\roman*)]
  \item The result in Theorem \ref{thm:min-max} is new to the kernel smoothing literature. The existing result on the optimal convergence rate by \citet{park2009} uses only a full rank smoother. Moreover, their result is for the Nadaraya-Watson smoother which is different from the projection-based smoother used in our study.
  \item The assumption such as  $m(\cdot) \in \mathcal{F}(\pi)$ and density $f(\cdot) \in \mathcal{F}(\pi-2)$ for $\pi \ge 4$ is common in the boosting literature. We refer to \citet{park2009} for similar condition (to be exact  $f(\cdot) \in \mathcal{F}(\pi-1)$) with Nadaraya-Watson smoother.
  \item  Boosting may adopt to higher order smoothness since it refits several times, as detailed in \citet{buhlmann2003boosting}. The optimal rate for $\pi=4$ is $n^{-8/9}$ for the bandwidth $n^{-1/9}$. This takes at least $1$ and $2$ iterations for local linear smoother $(p=1)$, $\bm{H}_1^*$,  and local constant smoother $(p=0)$, $\bm{H}_0^*$, respectively.   
\end{enumerate}

%
%
\subsection{Optimal Bandwidth and the number of iterations}\label{subsec:tuning} 
Despite the fact that boosting is resistant to overfitting \citep{buhlmann2003boosting}, selecting the optimal number of boosting iterations $(b)$  to avoid overfitting is crucial. It should also be noted that the learner's bandwidth $(h)$ influences whether the smoother is a strong or a weak learner. Boosting with a weak learner requires several iterations, but boosting with a  strong learner overfit very quickly. In some sense both $b$ and $h$ depend on each other.

By Theorem \ref{thm:min-max}, we have $h_{opt}=O(n^{-1/(2\pi+1)})$ which minimizes $E(MSE)$ for sufficiently large $b$. In general, the optimum $(b,h)$ will be found as the minimizers of 
\begin{align}
    P(b,h) &= \underset{h \in \mathcal{H}, b=1,\ldots,B}{\min} P_b(h)= \underset{h \in \mathcal{H},  b=1,\ldots,B}{\min} n^{-1}\sum_{i=1}^n \left( Y_i - \widehat{m}_{b}^*(X_i) \right)^2,
\end{align}
where $B$ is some fixed large number and $\mathcal{H}$ is a finite grid. Because of the overfitting problem associated with $P(b,h)$, in practice, either the cross-validation (CV) or the generalized cross-validation (GCV) approach is employed to tune the bandwidth $h$ and the number of boosting iterations, $b$. Using CV approach, we find the optimum $(b,h)$ as
\begin{align}
    CV(b,h) &= \underset{h \in \mathcal{H}, b=1,\ldots,B}{\min} CV_b(h) = \underset{h \in \mathcal{H}, b=1,\ldots,B}{\min} n^{-1}\sum_{i=1}^n \left( Y_i - \widehat{m}_{b,i}^*(X_i) \right)^2,
\end{align}
where $\widehat{m}_{b,i}^*(X_i)$ is estimated without involving $(X_i, Y_i)$.  Given the properties of $\bm{H}^*$ in Section \ref{sec:back}, we express the matrix $\mathcal{S}_b$  in $\widehat{\bm{m}}_b^*=\mathcal{S}_b \bm{y}$, given in (\ref{eqn:mbstar}), as
\begin{align}
    \mathcal{S}_b = \bm{I}-(\bm{I}-\bm{H}^*)^{b+1} = (b+1) \left(\bm{H}^*- \frac{b}{2} \bm{H}^{*^2} + o_p(1)\right). \label{eqn:sb-binom}
\end{align}
It is worth mentioning that the above expression (\ref{eqn:sb-binom}) also gives rise to an interesting relation between the boosting estimators at each iteration,
\begin{align*}
    \widehat{\bm{m}}_{b}^* &\approx (b+1) \widehat{\bm{m}}_{0}^*. 
\end{align*}
Since the smoother $\bm{H}^*$ has the diagonal elements of order $O(n^{-1}h^{-1})$ and the off-diagonal elements of order $O(n^{-1})$ \citep{hua08analysis}, by doing calculations similar to equation (2.5) in \citet{hardle1988far}, we can show that
\begin{align*}
    CV_b(h)/ P_b(h) \approx 1 + O_p(1/nh),
\end{align*}
for a fixed $b$, uniformly over $h$. We now define the average squared error at  iteration $b$ as
\begin{align}
    ASE_b(h) &= n^{-1}\sum_{i=1}^n (\widehat{m}_b^*(X_i) - m(X_i))^2.
\end{align}
Let $\widehat{h}_{0,b}$,  $\widehat{h}_{CV,b}$, and $h_{opt,b}$ be the minimizers of $ASE_b(h)$, $CV_b(h)$, and $E(MSE)$, respectively, at the $b$th iteration. Although giving a rigorous framework is beyond the scope of this paper, we conjecture that it is possible to obtain
\begin{align*}
 \widehat{h}_{CV,b}/ \widehat{h}_{0,b} \rightarrow 1, \qquad  \widehat{h}_{0,b}/h_{opt,b} \rightarrow 1,  
\end{align*}
in probability, by proceeding similar to \citet{ hardle1985optimal} and \citet{hardle1988far}. Since $h_{opt,b} \approx h_{opt}$ for sufficiently large $b$ (please see Theorem \ref{thm:min-max}), $\widehat{h}_{CV,b}$ is indeed a good choice.

Alternatively, given their asymptotic equivalence, we may use generalized cross-validation (GCV) approach to substitute the computationally intensive cross-validation procedure. The GCV at iteration $b$ is computed as
\begin{align}
    GCV &= \frac{n^{-1}\sum_{i=1}^n (Y_i- \widehat{m}_b^*(X_i))^2}{(1-\text{tr}(\mathcal{S}_b)/n)^2}. \label{eqn:gcv}
\end{align}
The properties of GCV have been investigated extensively in the literature. For recent results, we refer to \citet{amini2015optimal, roozbeh2018optimal} and the references therein.
%

\section{Contributions to Robustified Boosting}  \label{sec:rob-l2}

Because of the $L_2$ loss function, boosting is sensitive to outliers in the data. 
We robustify the boosting procedure developed in Section \ref{sec:l2boost}. The pseudo data technique detailed in \citet{cox1983asymptotics}, \citet{oh2007role} and \citet{oh2008recipe} is the key component of our methodology.

Let $\rho(x)$ be a convex function that is symmetric at zero grows slower than $x^2$ as $|x|$ becomes larger.  The Huber loss function is a well-known example of such a function, which is stated as  
\begin{align}
  \rho(x) &= \begin{cases}  x^2, & \vert x \vert \le c \\ 2c \vert x \vert-c^2, & \vert x \vert >c,  \end{cases} \label{eqn:huber}
\end{align}
where $c>0$ is the cutoff that is determined based on the data. Let $\psi= \rho'$ be the derivative of $\rho$. From (\ref{eqn:back-iuobj}), we observe that the estimator $\widehat{\bm{m}}^*=(\widehat{m}^*(X_1),\ldots, \widehat{m}^*(X_n))^T$ satisfies the following score equation under $L_2$ loss
\begin{align}
  n^{-1}\sum_{i=1}^n (Y_i-\widehat{m}^*(X_i)) &=0.
  \label{eqn:rob-l2score}
\end{align} 
Let $\widehat{\bm{m}}_\rho^*=(\widehat{m}_\rho^*(X_1), \ldots, \widehat{m}_\rho^*(X_n))^T$ be a robustified estimator for model (\ref{eqn:model}) under the loss function $\rho(\cdot)$. Therefore, from (\ref{eqn:rob-l2score}), it is reasonable to assume that the estimator $\widehat{\bm{m}}_\rho^*$ satisfies the following score equation 
\begin{align}
   n^{-1} \sum_{i=1}^n \psi(Y_i- \widehat{m}_\rho^*(X_i)) &=0.
   \label{eqn:rob-siscore}
\end{align}
Similar to \citet{oh2008recipe}, define the pseudo data
\begin{align*}
  Z_i= m(X_i) + \psi(\epsilon_i)/2, \qquad i=1,\ldots,n,
\end{align*}
given $\psi(\epsilon_i)$ exists and has a finite variance. 
By taking the empirical pseudo data $\widetilde{Z}_{\rho i}= \widehat{m}_\rho^*(X_i) + \psi\{ Y_i- \widehat{m}_\rho^*(X_i) \}/2$, we can express the score function (\ref{eqn:rob-siscore}) as 
\begin{align}
   n^{-1} \sum_{i=1}^n 2\{\widetilde{Z}_{\rho i}-\widehat{m}_\rho^*(X_i)\} &=0,
   \label{eqn:rob-tsiscore}
\end{align}
which is equivalent to the score equation under $L_2$ loss (\ref{eqn:rob-l2score}) with $\widetilde{Z}_{\rho i}$ as the response variable. This transformation serves as a motivation for the pseudo outcome approach described in \citet{cox1983asymptotics} and \citet{oh2007role}. It facilitates a theoretical analysis and provides an easily computing algorithm for model estimation \citep{cox1983asymptotics}. Since $\widetilde{Z}_{\rho i}$ involves $\widehat{m}_\rho^*(X_i)$ which is unknown, in practice, an iterative algorithm is needed. 

Let $\widetilde{\bm{m}}^*=(\widetilde{m}^*(X_1), \ldots, \widetilde{m}^*(X_n))^T$ denote an iterative estimator that satisfies
\begin{align}
  n^{-1} \sum_{i=1}^n 2\{\widetilde{Z}_i-\widetilde{m}^*(X_i)\} &=0,
  \label{eqn:rob-mtilde}
\end{align}
for $\widetilde{Z}_i=\widetilde{m}^*(X_i)+ \psi(Y_i-\widetilde{m}^*(X_i))/2$ for $i=1,\ldots,n$. Because it employs the $L_2$ loss function, the properties of $\widetilde{m}^*(X_i)$ are comparatively easy to acquire. The goal now is to prove that $\widetilde{m}^*(X_i)$ is asymptotically equivalent to $\widehat{m}_\rho^*(X_i)$ with the latter's properties remaining comparable to the former.

Let $\widetilde{\bm{z}}=(\widetilde{Z}_1,\ldots,\widetilde{Z}_n)^T$. We first present an algorithm for estimating $\widetilde{\bm{m}}^*$ and then provide a robustified boosting algorithm.
\vspace{1em}

\textbf{(II). Pseudo data Algorithm for $\widetilde{\bm{m}}^*$:}
\begin{enumerate}
  \item Obtain initial estimate $\widetilde{\bm{m}}_{(0)}^*=\bm{H}^*\bm{y}$.
  \item Set $\widetilde{\bm{z}}_{(0)}=\bm{y}$.
  \item Repeat the following steps ($k=1,\ldots$) until convergence 
  \begin{enumerate}
    \item Compute $\widetilde{\bm{z}}_{(k)}=\widetilde{\bm{m}}_{(k-1)}^*+ \psi(\bm{y}-\widetilde{\bm{m}}_{(k-1)}^*)$.
    \item Compute the estimator $\widetilde{\bm{m}}_{(k)}^*=\bm{H}^* \widetilde{\bm{z}}_{(k)}$.
  \end{enumerate}
  \item At convergence, take the final estimator as $\widetilde{\bm{m}}^* :=\widetilde{\bm{m}}_{(k)}^*$.
\end{enumerate}

The robustified $L_2$ boosting algorithm employs the above pseudo data algorithm (II) at every iteration. Mainly, it includes the following steps.

\vspace{1em}

\textbf{(III). Robustified $L_2$ Boosting Algorithm:}
\begin{itemize}
    \item[] Step 1 (Initialization): Given the data $\{X_i, Y_i \}$, $i=1,\ldots,n$, obtain an initial estimate $\widehat{\bm{m}}^*_0:= \widetilde{\bm{m}}^*$, from Algorithm II.
    \item[] Step 2 (Iteration): Repeat for $b=1,\ldots,B$.
    \begin{enumerate}
        \item Compute the residual vector $\bm{\delta} :=\bm{y}-\widehat{\bm{m}}^*_{b-1}$.
        \item Use Algorithm II with residuals as a response variable, $\bm{y}=\bm{\delta}$, to obtain $\widehat{\bm{m}}^*_\delta := \widetilde{\bm{m}}^*$.
        \item Update $\widehat{\bm{m}}^*_b:= \widehat{\bm{m}}^*_{b-1} + \widehat{\bm{m}}^*_\delta$.
    \end{enumerate} 
\end{itemize}
The concept of pseudo data has been successfully applied in the context of smoothing splines \citep{cox1983asymptotics, cantoni2001resistant} and wavelet regression \citep{oh2007role} to derive the asymptotic results as well as to facilitate a computation algorithm. The key result is that the robust smoothing estimator is asymptotically equivalent to a least-square smoothing estimator applied to the pseudo data. The proposed robustified boosting Algorithm (III) is general in the sense that it can be used with any other smoother considered in the study. In our numerical study, in addition to $\bm{H}^*$, we also employ this algorithm with Nadaraya-Watson and spline smoothers.

We now state the following assumption which is required for Theorem \ref{thm:rob-eq}. It is found in \citet{oh2007role}.

\begin{assumption} \label{as:psi}
The function $\psi$ has a continuous second derivative and satisfies $\underset{-\infty < t < \infty}{\sup}|\psi''(t| < \infty$. Assume $\psi$ is normalized such that $E[\psi(t)]=0$, $E[\psi'(t)]/2=1$, $\text{var}\{\psi(t)\} < \infty$, and  $\text{var}\{\psi'(t)\} < \infty$.
\end{assumption}

In Theorem \ref{thm:rob-eq}, we show that the estimators $\widehat{\bm{m}}_\rho^*$ and $\widetilde{\bm{m}}^*$ are asymptotically equivalent.
\begin{theorem}\label{thm:rob-eq}
  Suppose the Assumptions \ref{as:1-x}, \ref{as:2-ker-h},  and  \ref{as:psi} hold. Suppose $m(\cdot)$ is  four times continuously differentiable. Let $C_n=n^{-1}E\Vert  \widetilde{\bm{m}}^*- \bm{m}_\rho^* \Vert^2$,  where  $\Vert \cdot \Vert$ denotes the Euclidean norm, and assume that $C_n \rightarrow 0$ as $n \rightarrow \infty$. Then,
  \begin{align}
    n^{-1}\Vert \widetilde{\bm{m}}^*- \widehat{\bm{m}}_\rho^* \Vert^2 /\sqrt{C_n} \rightarrow 0
  \end{align} 
  in probability as $n \rightarrow \infty$.
\end{theorem}
 The proof is a special case of the proof from Theorem 1 of \citet{oh2007role}.  Therefore, we omit the proof. It has mainly three basic parts: finding uniform bounds on the score functions  (\ref{eqn:rob-siscore}) and (\ref{eqn:rob-tsiscore}), applying a fixed point argument, and evaluating pseudo data score function (\ref{eqn:rob-tsiscore}) at the robust estimator. Theorem \ref{thm:rob-eq} aids in establishing the asymptotic properties of the robustified boosting estimator $\widehat{\bm{m}}_b^*$ in Algorithm (III). It is sufficient to deduce the asymptotic properties of  $\widetilde{\bm{m}}^*$, which may be derived similarly to (\ref{eqn:h0-bias}).

 We note that obtaining the theoretical properties of the robustified boosting algorithm is not easy given the involvement of the nonlinear function $\psi$. We defer this for future research.

\section{Simulation Study }\label{sec:sim}
In this section, we use simulations to assess the finite sample performances of the proposed boosting methods. We investigate two scenarios: one with no outliers and the other with few outliers. All the computations were done using the software \texttt{Julia} \citep{Julia-2017} on a CentOS 7 machine. 

\subsection{Example 1: Without outliers} \label{subsec:numeric}
For this example, we mimic the simulation design in \citet{buhlmann2003boosting}. The following models are used to generate data:
\begin{align}
  \begin{split}
  (M1)~:~ Y_i &= 0.8X_i + \sin(6X_i) + \epsilon_i,\\
  (M2)~:~ Y_i &= 0.4\{ 3\sin(4\pi X_i) + 2 \sin(3\pi X_i)\}+ \epsilon_i,
  \end{split}
  \label{eqn:sim-ex-m1}
\end{align}
 for $i=1,\ldots,n$, where $\epsilon_i \sim N(0,2)$ and $X_i \sim U(-0.5,0.5)$ which is different from \citet{buhlmann2003boosting}.  The function $(M2)$ is taken from \citet{park2009}. Using both $\bm{H}_0^*$ (local constant) and $\bm{H}_1^*$ (local linear) smoothers, we employ $L_2$ boosting to estimate (\ref{eqn:sim-ex-m1}). A fixed grid of length 200 is utilized to approximate the integrals for both of them. Both the Nadaraya-Watson \citep{di2008boosting}  and the cubic spline \citep{buhlmann2003boosting} smoothers are considered for comparison.  For kernel and projection smoothers, we use Epanechnikov and Gaussian kernels.

To assess the predictive performance of the above methods, we compute their out-of-sample prediction errors as follows:
\begin{enumerate}
  \item For a sample of size $n$, for each of the three kernels and one spline methods, we perform a 5-fold cross-validation for each pair of values $(h,b)$  and $(\lambda,b)$, respectively. Identify the optimal pairs $(\widehat{h}, \widehat{b})$ and $(\widehat{\lambda}, \widehat{b})$ which minimize the MSE defined in (\ref{eqn:mse}).
  \item We now simulate 100 datasets of size $n$ and apply the above 4 boosting algorithms with their respective $(\widehat{h}, \widehat{b})$ and $(\widehat{\lambda},\widehat{b})$.
  \item Compute the average ASE values of the 100 datasets, where 
  \begin{align}
  ASE &=n^{-1}\sum_{i=1}^n \{m(X_i)-\widehat{m}^*_{\widehat{b}}(X_i)\}^2.
  \label{eqn:mse-t}
  \end{align} 
  
  \item Repeat the above steps (1--3) $10$ times and compute the average and standard deviation of the average $ASE$ values obtained in step (3).
\end{enumerate}
For the cross-validation in step 1, we perform search on a grid of length $40$ in the intervals $[0.1,4]$ and $[0,1000]$ for bandwidth $h$ and spline smoothing parameter $\lambda$, respectively. This search is performed for $5000$ boosting iterations.

The average $ASE$ values for model (M1) for increasing sample sizes $n=100, 200, 500, 1000$ are shown in Table \ref{tab:ex11}. The results are consistent across the four smoothers: $\bm{H}^*$ with $p=0$ and $p=1$, Nadaraya-Watson, and cubic smoothing spline. The Nadaraya-Watson smoother provided relatively larger $ASE$ values for the Epanechnikov kernel than the Gaussian kernel. The results from $\bm{H}^*$ (local constant and local linear) are comparable across both kernels and similar to the results from spline smoother. The results for model (M2) are shown in Table \ref{tab:ex12}. For model (M2), Nadaraya-Watson smoother produced slightly better results. Overall, the $\bm{H}^*$ smoothers performed well for both kernel functions and their results are comparable to the results from other smoothers.

\begin{table*}
  \begin{center}
  \begin{minipage}{\textwidth}
  \caption{Model $(M1)$: Average and standard deviations of 10 average $ASE$ values across 100 simulated datasets. LC: local constant, LL: local linear, NW: Nadaraya-Watson, SS: Smoothing Splines. Ep: Epanechnikov kernel, Ga: Gaussian Kernel}\label{tab:ex11}
  \begin{tabular*}{\textwidth}{@{\extracolsep{\fill}}cccccc@{\extracolsep{\fill}}}
  \toprule%
   Kernel & Sample size (n) & $\bm{H}^*_0$ (LC) & $\bm{H}^*_1$ (LL) & NW & SS  \\
  \midrule
  \multirow{5}{*}{Ep} & 100  &0.1117    &0.1174   &0.1596 &   0.1246 \\
  & & (0.045) & (0.069) & (0.075) & (0.053) \\
  & 200  &0.0691   &0.0646   &0.0728 &  0.0639 \\
  && (0.023) & (0.030)  & (0.037)  & (0.024) \\
  & 500  &0.0198   &0.0273  &0.0307  &0.0158 \\
  & &  (0.013) & (0.016) &  (0.025) & (0.009) \\
  & 1000  &0.0165  &0.0136   &0.0137  &0.0114\\
   & & (0.014)  & (0.005) & (0.007) & (0.007) \\
  \midrule
  \multirow{5}{*}{Ga}&100   &0.1341   &0.1249  &0.1273   &0.1246\\
  & & (0.057) & (0.072) & (0.055) & (0.053) \\
   &200  &0.0628  &0.0599   &0.0590   &0.0639\\
   & & (0.025)  & (0.026) &  (0.022) & (0.024) \\
   &500  &0.0190   &0.0228   &0.0231   &0.0158\\
   & & (0.018) & (0.018) & (0.028) & (0.009) \\
   &1000  &0.0119  &0.0111   &0.0107  &0.0114\\
   & & (0.009) & (0.005) & (0.006)  & (0.007)\\
  \bottomrule
  \end{tabular*}
  \end{minipage}
  \end{center}
  \end{table*}

\begin{table*}[h]
  \begin{center}
  \begin{minipage}{\textwidth}
  \caption{ Model $(M2)$: Average and standard deviations of 10 average $ASE$ values (\ref{eqn:mse-t}) for 100 simulated data for different sizes and using different kernel functions. LC: local constant, LL: local linear, NW: Nadaraya-Watson, SS: Smoothing Splines. Ep: Epanechnikov kernel, Ga: Gaussian Kernel}\label{tab:ex12}
  \begin{tabular*}{\textwidth}{@{\extracolsep{\fill}}cccccc@{\extracolsep{\fill}}}
  \toprule%
   Kernel & Sample size (n) & $\bm{H}^*_0$ (LC) & $\bm{H}^*_1$ (LL) & NW & SS  \\
  \midrule
  \multirow{5}{*}{Ep} & 100  &0.2233    &0.2055   &0.2316 &   0.3590 \\
  & & (0.090) & (0.148) & (0.086) & (0.090) \\
  & 200  &0.1231   &0.1395   &0.1351 &  0.1628 \\
  && (0.047) & (0.067)  & (0.041)  & (0.041) \\
  & 500  &0.0537   &0.0511  &0.0489  &0.035 \\
  & &  (0.036) & (0.028) &  (0.029) & (0.009) \\
  & 1000  &0.0299 &0.0305   &0.0282  &0.0178\\
   & & (0.016)  & (0.025) & (0.018) & (0.010) \\
  \midrule
  \multirow{5}{*}{Ga}&100   &0.1796   &0.1858  &0.1926   &0.3590\\
  & & (0.105) & (0.126) & (0.106) & (0.090) \\
   &200  &0.1045  &0.1024   &0.1045   &0.1628\\
   & & (0.027)  & (0.040) &  (0.031) & (0.041) \\
   &500  &0.0391   &0.0383   &0.0341   &0.0350\\
   & & (0.020) & (0.025) & (0.028) & (0.018) \\
   &1000  &0.0224  &0.0232   &0.0173  &0.0178\\
   & & (0.018) & (0.021) & (0.012)  & (0.010)\\
  \bottomrule
  \end{tabular*}
  \end{minipage}
  \end{center}
  \end{table*}

We also assess the predictive performance of  low-rank smoothers. The mean and standard deviation of 10 average $ASE$ values across 100 simulated datasets for smoother $\bm{H}_0^*$ (local constant) are shown in Table \ref{tab:ex1-rr}. We find that for samples of sizes 100 and 200, only ($d_n=$) 10  basis functions are needed to approximate the full-rank smoother. The approximation improves as $d_n$ increases. Furthermore, the findings show that the low-rank smoother outperforms the full-rank smoother in terms of prediction accuracy. The results for model (M2) are provided in Table \ref{tab:ex1-rr-m2}. The findings remain similar to those from Table \ref{tab:ex1-rr}.   

  \begin{table*}[h]
    \begin{center}
    \begin{minipage}{\textwidth}
    \caption{Model $(M1)$: Low-rank smoother $\bm{H}_0^*(d)$: Average and standard deviations of 10  average $ASE$ values (\ref{eqn:mse-t}) of 100 simulated datasets. $d_n$: number of basis functions (low-rank). }\label{tab:ex1-rr}
    \begin{tabular*}{\textwidth}{@{\extracolsep{\fill}}cccccc@{\extracolsep{\fill}}}
    \toprule%
    Sample size (n) & $d_n=2$ & $d_n=5$ & $d_n=10$ & $d_n=15$ & $d_n=n$  \\
    \midrule
     \multirow{2}{*}{100} & \textbf{0.1141} &0.1201    &0.115  &0.1149 &   0.1149 \\
         & (0.039) & (0.057) & (0.047) & (0.047) & (0.047) \\
      \multirow{2}{*}{200} & 0.0769 & \textbf{0.0664}   &0.0683   &0.0685&  0.0686 \\
        & (0.047) & (0.028) & (0.027)  & (0.0271)  & (0.0271) \\
      \multirow{2}{*}{500}  & \textbf{0.0171}  & 0.0172   &0.0191  &0.0223  &0.0320 \\
         & (0.0065) & (0.0080) & (0.0128) &  (0.0224) & (0.0387) \\
    \bottomrule
    \end{tabular*}
    \end{minipage}
    \end{center}
    \end{table*}

\begin{table*}[h]
  \begin{center}
  \begin{minipage}{\textwidth}
  \caption{Model $(M2)$: Low-rank smoother $\bm{H}_0^*(d)$: Average and standard deviations of 10 average $ASE$ values (\ref{eqn:mse-t}) of 100 simulated datasets. $d_n$: number of basis functions (low-rank). }\label{tab:ex1-rr-m2}
  \begin{tabular*}{\textwidth}{@{\extracolsep{\fill}}cccccc@{\extracolsep{\fill}}}
  \toprule%
  Sample size (n) & $d_n=2$ & $d_n=5$ & $d_n=10$ & $d_n=15$ & $d_n=n$  \\
  \midrule
   \multirow{2}{*}{100} & 0.2248 &  \textbf{0.2157}    & 0.2163  & 0.2163 &   0.2163 \\
       & (0.1109) & (0.1049) & (0.1102) & (0.1101) & (0.1101) \\
    \multirow{2}{*}{200} & 0.1470 &  0.1338   & \textbf{0.1053}   &0.1039 &  0.1039 \\
      & (0.0460) & (0.0608) & (0.0284)  & (0.0291)  & (0.0291) \\
    \multirow{2}{*}{500}  &  0.1066  & 0.0406   & \textbf{0.0394}  &0.0451  &0.0507 \\
       & (0.0257) & (0.0161) & (0.0209) &  (0.0266) & (0.0366) \\
  \bottomrule
  \end{tabular*}
  \end{minipage}
  \end{center}
  \end{table*}

\subsection{Example 2: With outliers} \label{subsec:numeric-rob}
In this section, we evaluate the performance of the robustified $L_2$ boosting algorithm. The same simulation design of Section \ref{subsec:numeric} is considered with one change; errors $\epsilon_i$ are now simulated from a $t$-distribution with $3$ degrees of freedom. We use Huber loss function (\ref{eqn:huber}) to robustify the boosting algorithm. The results of non-robust boosting algorithms are also provided for comparison. To compute the average $ASE$ values, we used the same approach as in Example 1. However, because of the Huber loss function, the mean square error values and cross-validation errors for the boosting iterations are computed using the following function
\begin{align}
 RoMSE_{\rho} &= \frac{MSE}{n} \sum_{i=1}^n \rho\left(\frac{Y_i- \widehat{m}_b^*(X_i)}{MSE^{1/2}}\right), \label{eqn:romse} 
\end{align}
where $MSE$ is defined in (\ref{eqn:mse}). The Huber constant $c$ is a tuning parameter that needs to be estimated from the data. In the following section where we analyze a real data, we consider the choice of $\widehat{c}=1.345\widehat{\sigma}$, where $\widehat{\sigma}$ is any robust estimator of the population standard deviation $\sigma$ \citep{huber2004robust}.  

Table \ref{tab:ex21} shows the results of the robustified boosting for the Huber constants $c=1$ and $c=2$. The robustified boosting approach produced lower $ASE$ values than the non-robust ones. This implies that the robustified methods minimize the effect of outliers in the boosting procedure. Surprisingly, $\bm{H}^*$ smoothers outperformed  the Nadaraya-Watson smoother in terms of smaller errors. The results from $\bm{H}_0^*$ (local constant) and spline smoother are nearly identical. Furthermore, we find that the choice of the Huber constant $c$ is not as important for large samples. 

Overall,  the projection-based smoothers $\bm{H}^*$ are demonstrated to be very useful tools for the boosting algorithm. They are useful for investigating the effect of low-rank smoothers on boosting because of their appealing theoretical properties. Moreover, the robustified boosting approach outperforms the original $L_2$ boosting with the contaminated data. 
%
\begin{sidewaystable}
  \begin{center}
  \begin{minipage}{\textwidth}
  \caption{Average and standard deviations of 10 average $ASE$ values (\ref{eqn:mse-t}) for 100 simulated data of different sizes and using different Huber constant $c$ values. LC: local constant, LL: local linear, NW: Nadaraya-Watson, SS: Smoothing Splines.}\label{tab:ex21}
  \begin{tabular*}{\textwidth}{@{\extracolsep{\fill}}cccccccccc@{\extracolsep{\fill}}}
  \toprule%
   \multirow{2}{5em}{Huber constant($c$)} & \multirow{2}{3em}{Sample size(n)} & \multicolumn{2}{c}{$\bm{H}^*_0$ (LC)} & \multicolumn{2}{c}{$\bm{H}^*_1$ (LL)} & \multicolumn{2}{c}{NW} & \multicolumn{2}{c}{SS} \\
   \cline{3-10}
   &    &   robust & non-robust & robust & non-robust & robust & non-robust & robust & non-robust\\  
  \midrule
  \multirow{5}{*}{1.0} & 100  &0.1028    &0.2018  &0.1101 &  0.3954 & 0.1517 & 0.2150 &  0.1179 &  0.2183\\
                        & & (0.081) & (0.168) & (0.091) & (0.297) & (0.088) & (0.095) &  (0.098) & (0.157)\\
                     & 200  &0.0428   &0.0976   & 0.0550 & 0.2350 &  0.0508 &  0.0941 & 0.0364  & 0.0957\\
                          && ( 0.024) & (0.061)  & (0.051)  & (0.239) & (0.027) & (0.063) & (0.027) & (0.068)\\
                      & 500  &0.0199   &0.0325  & 0.0254 & 0.1530 & 0.0207  & 0.0469  & 0.0172   & 0.0327\\
                        & &  (0.009) & (0.016) &  (0.018) & (0.071)  & (0.011) &  (0.015) & (0.010) & (0.024)\\
   \midrule
  \multirow{6}{*}{2.0}& 100   &0.1309   &0.2018  &0.1158   &0.3954 & 0.1981  & 0.2150 & 0.1452 & 0.2183\\
                          & & (0.111) & (0.168) & (0.083) & (0.297) & (0.098) & (0.095) & (0.109) & (0.157)\\
   &200  &0.0460  &0.0976   &0.0617   &0.2350 & 0.0531 & 0.0941 & 0.0433 & 0.0957 \\
      & & (0.028)  & (0.061) &  (0.058) & (0.239) & (0.031) & (0.063) & (0.035) & (0.068) \\
   &500  &0.0203   &0.0325   &0.0237   &0.1530 & 0.0215 & 0.0469 & 0.0178 & 0.0327\\
   & & (0.009) & (0.016) & (0.014) & (0.071)  & (0.010) & (0.015) & (0.009) & (0.024)\\
   \bottomrule
  \end{tabular*}
  \end{minipage}
  \end{center}
\end{sidewaystable}

%
\section{Real Application}\label{sec:real-app}
We consider the \texttt{cps71} data from \citet{ullah1985specification}, \citet{pagan1999nonparametric} and \citet{hayfield2008nonparametric} (1971 Canadian Public Use Tapes). It includes age and income information for 205 Canadian individuals who were educated to the thirteenth grade. Schooling for these individuals was believed to be equal across both genders. 

We are interested in the following model:
\begin{align}
  \ln(wage_i) &= m(age_i) + \epsilon_i, \qquad i=1,\ldots, 205.
  \label{eqn:cps-model}
\end{align}
Figure \ref{fig:sps-scatter} depicts the relationship between the covariate age and the outcome variable which is the logarithm of wage. Their relationship appears to be quadratic \citep{pagan1999nonparametric}, with multiple outliers, particularly among the elderly. This data is used to test the proposed robust and non-robust boosting methods. We use $200$ grid points to calculate smoother matrices $\bm{H}^*$ (both local constant and local linear). The value for the Huber's constant $c$ is chosen as, $c=1.345~\widehat{\sigma}$, where $\widehat{\sigma}$ is a robust location-free scale estimate $\sigma$ \citep{rousseeuw1993alternatives} based on the suggestion of \citet{oh2007role} and \citet{huber2004robust}. The \emph{gam} function in \texttt{mgcv} package \citep{wood2015package} is used to estimate model (\ref{eqn:cps-model}), and the residuals are used to compute $\widehat{\sigma}$.  The calculation are performed on a  \texttt{CentOS} machine using the  programming language \texttt{Julia} \citep{Julia-2017}.
\begin{figure}[h]
  \centering
  \includegraphics[scale=0.6]{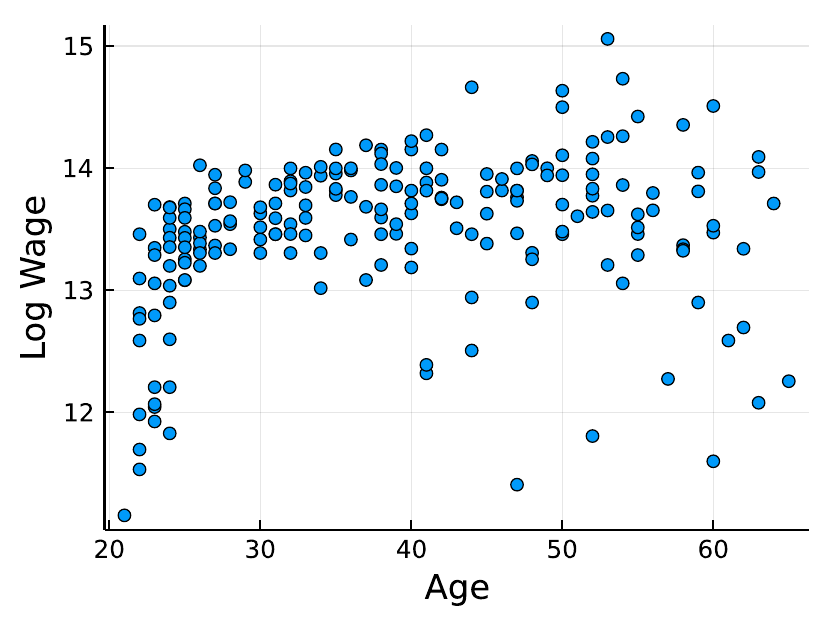}
  \caption{Scatter plot between the covariate Age and the response variable Log Wage for \texttt{cps71} data.}
  \label{fig:sps-scatter}
\end{figure}

Smoothers $\bm{H}^*$ (local constant, local linear), Nadaraya-Watson, and cubic spline are used to estimate the model \ref{eqn:cps-model} utilizing  Algorithms (I) and (III) for non-robust and robust methods, respectively. We use Epanechnikov kernel  for kernel smoothing methods. To select the  optimal pair  $(\widehat{h}, \widehat{b})$ for each smoother and algorithm combination, a 5-fold cross-validation is employed. Grid search is performed for the bandwidth range $h=[5,20]$, smoothing parameter range $\lambda=[0,5000]$, and the number of boosting iterations $b=1000$. 

 On the whole dataset of 205 observations, we train both robust and non-robust boosting approaches. For evaluation, we create a trimmed data which excludes observations when logarithm wage lies outside of (14.15,13.03) with age more than 30 years. This approximately excludes 20\% observations. The model evaluation procedure is as follows:
 \begin{itemize}
     \item Using the whole data, perform a 5 fold cross-validation or a GCV to find the optimal parameters $(h,b)$ or $(\lambda,b)$. The cross-validation uses $MSE$ (in \ref{eqn:mse}) for non-robust methods and uses $RoMSE_{\rho}$ (in \ref{eqn:romse}) for robust methods. Alternatively, the GCV approach uses $GCV$ (in \ref{eqn:gcv}) for non-robust methods and $RoGCV$ for robust methods which,  at iteration $b$, is defined as
     \begin{align*}
         RoGCV &= \frac{n^{-1} \sum_{i=1}^n \rho(Y_i- \widehat{m}_b^*(X_i))}{\left( 1 - \text{tr}(\mathcal{S}_b)/n\right)^2}.
     \end{align*}

     \item For the optimal values of $(\widehat{h}, \widehat{b})$ or $(\widehat{\lambda}, \widehat{b})$, compute the fitted values for the trimmed data and calculate MSE values. 
 \end{itemize}

First, we choose the cross-validation procedure. The results from all the above mentioned smoothers are presented in Table \ref{tab:real-ex} and Figure \ref{fig:cps-optd}. We find that all the robust methods performed better than their non-robust counterparts. Unlike the results from Table \ref{tab:ex21} where $\bm{H}_0^*(LC)$ produced the smallest $ASE$ values, the NW smoother produced the smallest $MSE$ values in Table \ref{tab:real-ex}. We believe that this is due to the smaller bandwidth selected for  NW smoother in Table \ref{tab:real-ex} which can also be observed from Figure \ref{fig:cps-optd}. Additionally, we find that the results are not very sensitive to the value of Huber's constant $c$. This is much desirable in practice. 
\begin{figure}[h]
  \centering
  \includegraphics[scale=0.3]{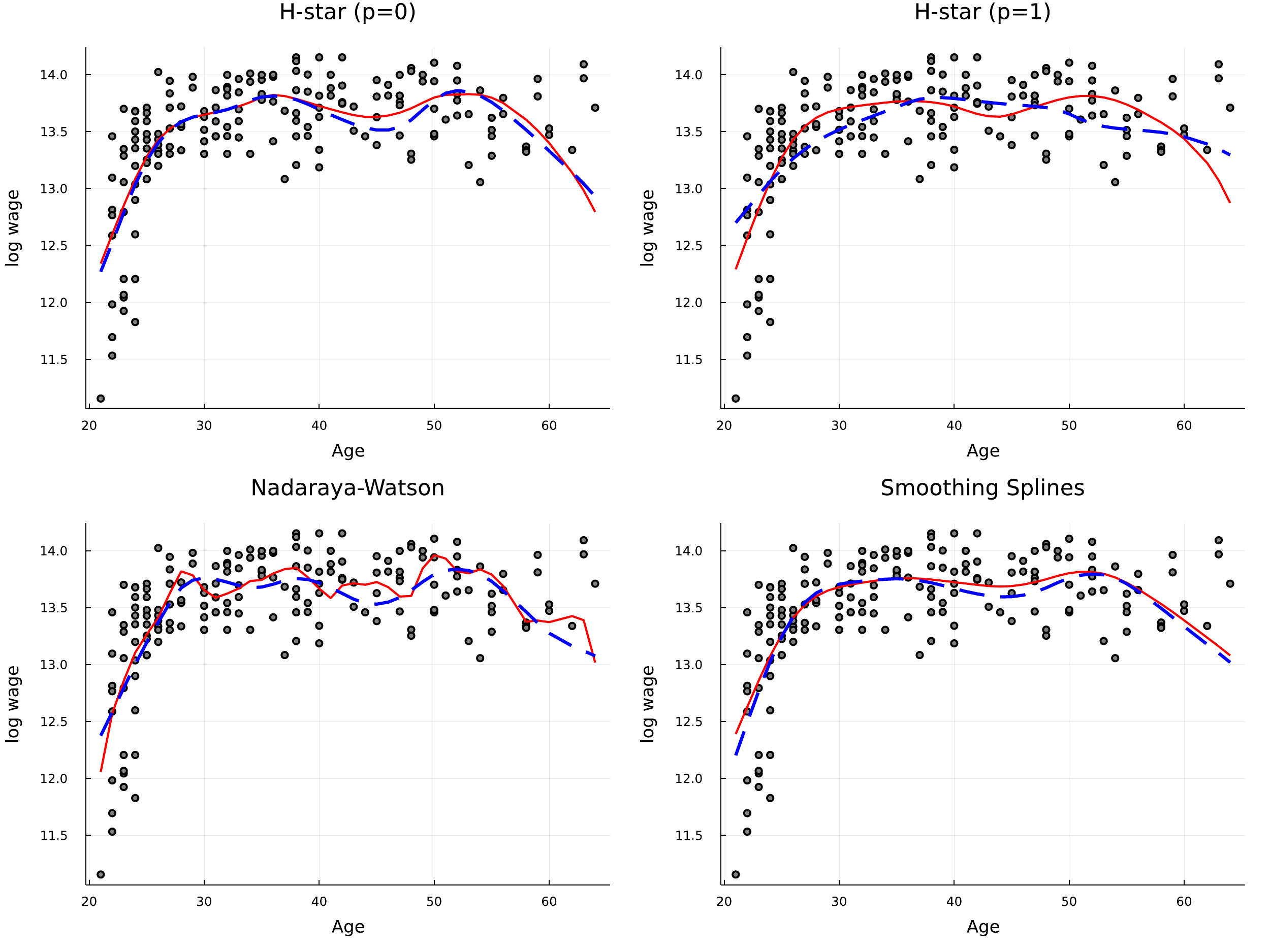}
  \caption{Robust (solid) and non-robust (dash) fits for the model (\ref{eqn:cps-model}) on full data. Here, Huber constant $\widehat{c}=\widehat{\sigma}=0.64$. The optimal pair $(\widehat{h},\widehat{b})$ for each model is chosen based on the 5-fold cross-validation. The dots denote the observations from the trimmed data. }
  \label{fig:cps-optd}
\end{figure}
\begin{table}[h]
\scriptsize
  \begin{center}
  \begin{minipage}{\textwidth}
  \caption{Results based on optimal $(\widehat{h}, \widehat{b})$ and  $(\widehat{\lambda}, \widehat{b})$ chosen by the cross-validation. Here $\widehat{\sigma}=0.64$. The $MSE$ (\ref{eqn:mse}) and $CV$ values are computed based on the trimmed data and on the whole data, respectively. }\label{tab:real-ex}
  \begin{tabular*}{\textwidth}{@{\extracolsep{\fill}}cccccccccc@{\extracolsep{\fill}}}
  \toprule%
  \multirow{2}{6em}{Huber constant ($\widehat{c}$)} &\multirow{2}{5em}{Smoother} & \multicolumn{4}{c}{Robust} & \multicolumn{4}{c}{Non-robust} \\
   \cline{3-10}
    &   &     $\widehat{h}$ or $\widehat{\lambda}$  & $\widehat{b}$ & $MSE$ & $CV$  & $\widehat{h}$ or $\widehat{\lambda}$ & $\widehat{b}$ & $MSE$ & $CV$ \\  
  \midrule
  \multirow{2}{5em}{$1.34\widehat{\sigma}$} &$\bm{H}_0^*$  &  13.84 & 330    & 0.1418  & 0.1057 &12.69 &  255 & 0.1440   & 0.2976\\
  &$\bm{H}_1^*$  &  20.00 & 227    & 0.1351 &0.1066 & 15.00 &  79  & 0.1491 & 0.2995\\
  &NW            &  5.76 & 19     & 0.1344 & 0.0938   &6.15 &  20  & 0.1443 & 0.2560 \\
  &SS            &  128.20 & 1    & 0.1348 & 0.1070   & 256.41 & 2 & 0.1367 & 0.3001\\
  \midrule
  \multirow{2}{5em}{$\widehat{\sigma}$}&$\bm{H}_0^*$  &  13.84 & 445    & 0.1440 & 0.0899 &12.69 &  255 & 0.1440 & 0.2976   \\
  &$\bm{H}_1^*$  &  20.00 & 219    & 0.1380 & 0.0911 & 15.00 &  79  & 0.1491 & 0.2995\\
  &NW            &  5.76 & 21     & 0.1347 &0.0816   &6.15 &  20  & 0.1443 & 0.2560 \\
  &SS            &  128.20 & 1    & 0.1362 & 0.0915   & 256.41 & 2 & 0.1367 & 0.3001\\
  \midrule
  \multirow{2}{5em}{$1.6\widehat{\sigma}$} & $\bm{H}_0^*$  &  13.84 & 305  & 0.1417  & 0.1057 &12.69 &  255 & 0.1440   & 0.2976\\
  &$\bm{H}_1^*$  &  20.00 & 231   & 0.1354 & 0.1155 & 15.00 &  79  & 0.1491 & 0.2995\\
  &NW            &  5.76 & 19     & 0.1346 & 0.1014   &6.15 &  20  & 0.1443 & 0.2560 \\
  &SS            &  256.41 & 3    & 0.1341 & 0.1160   & 256.41 & 2 & 0.1367 & 0.3001\\

  \bottomrule
  \end{tabular*}
  \end{minipage}
  \end{center}
\end{table}
%

\begin{table}[h]
\scriptsize
  \begin{center}
  \begin{minipage}{\textwidth}
  \caption{Results based on optimal $(\widehat{h}, \widehat{b})$ and  $(\widehat{\lambda}, \widehat{b})$ chosen by the GCV. Here $\widehat{\sigma}=0.64$. The $MSE$ (\ref{eqn:mse}) and $GCV$ values are computed based on the trimmed data and on the whole data, respectively. (**) GCV score for NW smoother is approaching to zero with increasing $b$.}\label{tab:real-ex-gcv}
  \begin{tabular*}{\textwidth}{@{\extracolsep{\fill}}cccccccccc@{\extracolsep{\fill}}}
  \toprule%
  \multirow{2}{6em}{Huber constant ($\widehat{c}$)} &\multirow{2}{5em}{Smoother} & \multicolumn{4}{c}{Robust} & \multicolumn{4}{c}{Non-robust} \\
   \cline{3-10}
    &   &     $\widehat{h}$ or $\widehat{\lambda}$  & $\widehat{b}$ & $MSE$ & $GCV$  & $\widehat{h}$ or $\widehat{\lambda}$ & $\widehat{b}$ & $MSE$ & $GCV$ \\  
  \midrule
  \multirow{2}{5em}{$1.34\widehat{\sigma}$} &$\bm{H}_0^*$  &  12.69 & 311    & 0.1425 & 0.1328 & 12.69 &  300 & 0.1432   & 0.2913\\
  &$\bm{H}_1^*$  &  15.38 & 99    & 0.1343 &0.1325 & 15.76 &  89  & 0.1520 & 0.2912\\
  &NW(**)            &  5 & 1000     & 4372.00 & 0   & 14.61 &  265  & 0.1377 & 0.0136 \\
  &SS            &  256.41 & 2    & 0.1349 & 0.1329   & 256.41 & 2 & 0.1367 & 0.2925\\
  \midrule
  \multirow{2}{5em}{$\widehat{\sigma}$}&$\bm{H}_0^*$  &  13.46 & 281    & 0.1456 & 0.1213 & 12.69 &  300 & 0.1432 & 0.2913   \\
  &$\bm{H}_1^*$  &  15.38 & 97    & 0.1353 & 0.1214 & 15.76 &  89  & 0.1520 & 0.2912\\
  &NW(**)            &  5 & 1000     & 2406.3 & 0   & 14.61 &  265  & 0.1377 & 0.0136 \\
  &SS            &  128.20 & 1    & 0.1362 & 0.1217   & 256.41 & 2 & 0.1367 & 0.2925\\
  \midrule
  \multirow{2}{5em}{$1.6\widehat{\sigma}$} & $\bm{H}_0^*$  &  12.69 & 327  & 0.1420  & 0.1375 & 12.69 &  300 & 0.1432   & 0.2913\\
  &$\bm{H}_1^*$  &  15.38 & 91   & 0.1350 & 0.1372 & 15.76 &  89  & 0.1520 & 0.2912\\
  &NW(**)            &  5 & 1000     & 6153.28  & 0   &14.61 &  265  & 0.1377 & 0.0136 \\
  &SS            &  256.41 & 2    & 0.1352 & 0.1377   & 256.41 & 2 & 0.1367 & 0.2925\\

  \bottomrule
  \end{tabular*}
  \end{minipage}
  \end{center}
\end{table}

The results for the GCV approach are provided in Table \ref{tab:real-ex-gcv}. We observe that the GCV criterion for the NW smoother fail to  provide correct results. Given the GCV formulation, we attribute this issue to the NW smoother which lacks some appealing properties of $\bm{H}_0^*$ and $\bm{H}_1^*$ smoothers. For the remaining three smoothers, the results from this table remain similar to the results from Table \ref{tab:real-ex}. Because GCV is not as computationally costly as cross-validation, this finding further illustrates the usefulness of the $\bm{H}^*$ smoothers. 

In conclusion, the proposed robust boosting approaches outperform their non-robust counterparts.

\section{Summary and Conclusions} \label{sec:sc}
We present a novel kernel regression based $L_2$ boosting approach to estimate a univariate nonparametric model. In the context of $L_2$ boosting, the suggested method overcomes the shortcomings of existing kernel-based methods. The theory established for spline smoothing is easily applicable since the smoother $\bm{H}^*$ utilized in the study is symmetric and has eigenvalues in the range $(0,1]$.  Simultaneously, specific results for kernel smoothing may also be obtained. We consider low-rank smoothers instead of the original smoother in our asymptotic framework. Low-rank smoothers make the boosting algorithm scalable to big datasets. In addition to the computational gains, based on our numerical results, we find that the low-rank smoothers may outperform the full-rank smoothers in terms of test data prediction error. Furthermore, we also robustify the proposed boosting procedure to alleviate the effect of outliers.

The present study considers  $L_2$ boosting for a univariate nonparametric model. The development of the $L_2$ boosting procedure for additive models is an intriguing area for future research. This topic may be very helpful in practice because the boosting approach also does variable selection.  

\section*{Acknowledgements}
The author would like to thank Prof. Li-Shan Huang for introducing projection smoothers to him and Dr. Abhijit Mandal for suggesting an important reference on robust smoothing.

\bibliographystyle{agsm} 
\bibliography{kboostref}
\end{document}